\documentclass[journal,onecolumn,12pt,draftclsnofoot]{IEEEtran}
\usepackage{graphicx}                                      
\usepackage{amssymb,amsthm,bm}  
\usepackage[cmex10]{amsmath} 
\usepackage{bbm}
\usepackage{multirow}   
\usepackage{subfig}
\usepackage{caption}

\usepackage{enumitem}                                      
\usepackage{algorithm}
\usepackage{algorithmic}  

\usepackage{diagbox}
\usepackage{color}
\usepackage{xcolor}

\usepackage{tikz} 
\usetikzlibrary{automata,positioning,chains,fit,shapes,calc}
\usetikzlibrary{arrows}
\usetikzlibrary{matrix}
\usetikzlibrary{arrows.meta}

\usepackage{caption}
\usepackage{booktabs}
\usepackage{cite}

\makeatletter
\newcommand{\mathleft}{\@fleqntrue\@mathmargin0pt}
\newcommand{\mathcenter}{\@fleqnfalse}
\makeatother              
\newtheorem{theorem}{Theorem}
\newtheorem{lemma}{Lemma}
\newtheorem{corollary}[lemma]{Corollary}

\newtheorem{definition}{Definition}

\newcommand{\utag}[2]{\mathop{#2}\limits^{\text{(#1)}}}
\newcommand{\uref}[1]{(#1)}

\newcommand{\Ie}{\textit{i.e., }}
\newcommand{\Eg}{\textit{e.g., }}

\long\def\symbolfootnote[#1]#2{\begingroup
\def\thefootnote{\fnsymbol{footnote}}\footnote[#1]{#2}\endgroup}

\newcommand{\cE}{\mathcal{E}}

\newcommand{\cH}{\mathcal{H}}

\newcommand{\cK}{\mathcal{K}}

\newcommand{\cO}{\mathcal{O}}

\newcommand{\cX}{\mathcal{X}}

\usepackage{mathrsfs}

\newcommand{\pb}[1]{p\left(#1\right)}
\newcommand{\wb}[1]{w\left(#1\right)}
\newcommand{\qb}[1]{q\left(#1\right)}

\newcommand{\bX}{\mathbf{X}}
\newcommand{\bY}{\mathbf{Y}}

\newcommand{\bx}{\mathbf{x}}
\newcommand{\by}{\mathbf{y}}

\newcommand{\revise}[1]{{\color{blue}{#1}}}

\begin{document}
\title{
Mechanisms for Hiding Sensitive Genotypes with Information-Theoretic Privacy
}

\author{%
  \IEEEauthorblockN{Fangwei Ye\IEEEauthorrefmark{1}, Hyunghoon Cho\IEEEauthorrefmark{2}, Salim El Rouayheb\IEEEauthorrefmark{1}%
  \thanks{A preliminary version of this paper was presented at IEEE International Symposium on Information Theory, Los Angeles, CA, USA, 2020. 
} 
}

  \IEEEauthorblockA{\IEEEauthorrefmark{1}%
   Rutgers University, Email: \{fangwei.ye, salim.elrouayheb\}@rutgers.edu}
  
  \and

  \IEEEauthorblockA{\IEEEauthorrefmark{2}%
Broad Institute of MIT and Harvard,                   
              Email: hhcho@broadinstitute.org} 
}

\maketitle

\begin{abstract}

Motivated by the growing availability of personal genomics services, we study an information-theoretic privacy problem that arises when sharing genomic data: a user wants to share his or her genome sequence while keeping the genotypes at certain positions hidden, which could otherwise reveal critical health-related information. A straightforward solution of erasing (masking) the chosen genotypes does not ensure privacy, because the correlation between nearby positions can leak the masked genotypes. We introduce an erasure-based privacy mechanism with perfect information-theoretic privacy, whereby the released sequence is statistically independent of the sensitive genotypes. Our mechanism can be interpreted as a locally-optimal greedy algorithm for a given processing order of sequence positions, where utility is measured by the number of positions released without erasure. We show that finding an optimal order is NP-hard in general and provide an upper bound on the optimal utility. For sequences from hidden Markov models, a standard modeling approach in genetics, we propose an efficient algorithmic implementation of our mechanism with complexity polynomial in sequence length. Moreover, we illustrate the robustness of the mechanism by bounding the privacy leakage from erroneous prior distributions. Our work is a step towards more rigorous control of privacy in genomic data sharing.

\end{abstract}

\section{Introduction}
\subsection{Motivation}
The rise of personal genomics, whereby private individuals are exposed to an increasing range of direct-to-consumer services for sequencing, sharing, or analyzing their genomes, is leading to growing concerns for genomic privacy~\cite{Hubaux17,Grishin19,Berger19}. A personal genome is a rich trove of information about the underlying individual, including predictors for disease risks and other health-related traits, which holds great potential for improving one's health, yet may cause harm if used against the individual. Unlike other types of personal data like passwords, one's genetic data cannot be replaced once leaked, and a data breach may even affect the relatives of the individual whose genome is leaked. In order to facilitate the sharing of genomes to improve public health and advance science, we need principled strategies for controlling the privacy risks associated with genomic data sharing.

A key need in this regard is to selectively limit the leakage of information about biological or health-related traits of an individual that can be inferred from the shared genetic data. For example, one may wish to hide certain \emph{genotypes} (an individual's genetic information at specific genomic positions) with well-established disease association before sharing his or her data with others (\Eg analytic service providers or researchers). Such a capability would give the individuals more fine-grained control over their genomic privacy.

A simple approach to privacy protection, whereby specific positions in the genome deemed sensitive by the individual are masked before sharing the data, does not provide sufficient privacy protection. This is because the correlation structure among nearby genomic positions induced by the biological processes of genetic inheritance can be used to reconstruct the masked data as demonstrated in a number of studies~\cite{nyholt2009,Von18}.
To prevent such an attack, one could alternatively erase all positions that are highly correlated with the sensitive sites~\cite{Gursoy19}, which may be achieved by masking the data within a large window around each sensitive position.
Unfortunately, depending upon the chosen size of window, these approaches either provide incomplete privacy protection or require an excessive amount of data to be erased in order to achieve strong privacy (as we demonstrate in our results), thus limiting the usefulness of the shared data.
Here, we aim to design a principled and effective mechanism for sharing a personal genome that provably hides sensitive positions,
while introducing a small amount of erasure.
Our techniques build upon the recent work on ON-OFF privacy~\cite{ONOFF_ISIT,ONOFF_ITW} while extending the theory to general data distributions beyond Markov chains addressed in the previous work.

It is worth noting that information-theoretic approaches are being increasingly explored for a diverse range of applications in genomics, including sequencing~\cite{Tse_2013}, genome-wide association study (GWAS)~\cite{GWS_Maddah-Ali_2018,cho2018secure}, genome assembly~\cite{Ilan,Sriiram_2017}, regulatory network of gene interactions (RNGI)~\cite{Milenkovic_coding_info},
and DNA-based information storage~\cite{Oligca_2016}. There are also recent works addressing the issue of genomic privacy, including a solution for private shotgun sequencing~\cite{Private_Maddah_2019} based on the intensively researched private information retrieval (PIR) problems~\cite{Sun_2017,Freij-Hollanti_2017,Banawan_2018,Tajeddine_2016,Li_side,Kadhe_2017} and differential privacy mechanisms for sharing aggregate genomic data~\cite{SIMMONS201654,fienberg2011privacy,cho2020privacy}. 
Broadly, our work can be viewed as a continuation of these efforts to develop effective genomic data processing tools from an information-theoretic perspective, yet for a novel problem that we introduce, \Ie the design of mechanisms for selectively hiding sensitive positions in genetic sequences.

\subsection{Genetics background}

An individual’s genome consists of a pair of sequences, one from each parent, each consisting of around 3 billion nucleotides (\textsf{A}, \textsf{C}, \textsf{G}, and \textsf{T}). Each sequence is referred to as a \emph{haplotype}. Since most of the genome sequence is identical between different individuals, a common way to compactly represent a personal genome is as a list of positions of variation, paired with the observed nucleotide(s) in the given individual (referred to as a \emph{genotype}). In this work, we consider the problem of sharing a list of genotypes corresponding to a \emph{single} haplotype of an individual. Although standard sequencing or genotyping pipelines produce a genotype at each position that convolves the two haplotypes, well-established methods exist~\cite{loh2016reference,browning2011haplotype}  for resolving this ambiguity in order to separate the two haplotypes (a process called \emph{phasing}), after which each haplotype could be individually considered.

In the setting of our work, we consider an adversary whose goal is to infer the target individual’s genotypes at specific positions in the genome, given a partially masked genetic sequence of the individual. In principle, this reconstruction task is equivalent to an extensively studied problem in bioinformatics known as \emph{genotype imputation}, originally developed for coping with the presence of missing data in the existing experimental pipelines for characterizing personal genomes. If one were to mask only the sensitive positions before sharing the data, existing imputation algorithms are expected to be effective at revealing the hidden genotypes using other genotypes in their respective neighborhoods.

A state-of-the-art algorithm for genotype imputation, Minimac~\cite{Das16},
is based on a classical model of genetic sequences introduced by Li and Stephens~\cite{Li03}.
In this model, a person’s genetic sequence is modelled as a mosaic of a large group of reference sequences from other individuals.
This model intuitively captures the underlying biological process of \emph{recombination}, which describes the interleaving of two haplotypes of each parent when their genetic material is passed onto the child.
Formally, these models are expressed as hidden Markov models (HMMs), where a sequence of genotypes of an individual is generated from a sequence of hidden states indicating which reference haplotype to copy the genotype from, for each corresponding position.
The parameters of these models are typically inferred from a large reference panel including tens of thousands of sequenced human genomes~\cite{mccarthy2016}.
Although alternative approaches to imputation (e.g. based on matrix factorization~\cite{chi2013}) exist, in our work we are especially interested in HMMs as the primary means to model the distribution of genotypes, considering the wide adoption of HMMs in genetics not only for imputation, but also for other standard tasks like phasing~\cite{browning2011haplotype} and simulation~\cite{dutheil2009ancestral}.
Further details of this model is provided in Section~\ref{sec:HMM}.

\subsection{Setup and contributions}

In this paper, we formulate the \emph{genotype hiding} problem: We consider a user who wishes to share a partially erased version of their genetic sequence while protecting a list of sensitive positions.
Privacy is measured by the mutual information between the sensitive positions and the released sequence, and we adopt a stringent privacy requirement that enforces zero mutual information (i.e., perfect privacy).
The goal of the problem is to design a privacy mechanism that satisfies this requirement, while minimizing the number of erasures introduced so as to maximize the utility of the data. 

We present such a mechanism with perfect privacy and provide a range of theoretical insights into its performance with respect to its utility, measured by the erasure rate.
The proposed mechanism sequentially processes the positions in the sequence in a given ordering and determines a suitable erasure rate at each position based on the previously released positions and the data generating distribution.
We prove that our mechanism can be viewed as a locally-optimal, greedy solution for minimizing the erasure rate at each position.
Furthermore, we give a lower bound on the number of erasures required for any mechanism satisfying the privacy constraint, and show that our privacy mechanism is in fact (globally) optimal for a class of data generative distributions defined by Markov chains.
We also show that finding the optimal ordering for the sequential mechanism is generally intractable (NP-hard), illustrating the limits of current techniques. 
Lastly, we derive an upper bound on potential privacy leakage due to inaccuracies in the estimation of the data generative model, suggesting that our mechanism is relatively robust to a small amount of noise in the data distribution.

For practical applications, we are particularly interested in data generating distributions induced by hidden Markov models (HMMs), which are broadly adopted in genetics as described in Section~\ref{sec:HMM}.
To this end, we also present a computationally-efficient algorithm to implement the proposed privacy mechanism based on HMMs, and provide an empirical evaluation of its performance on simulated datasets.

The rest of this paper is organized as follows. In Section~\ref{sec:formulation}, we formalize the genotype-hiding problem. Performance bounds are summarized in Section~\ref{sec:bound}. In Section~\ref{sec:mechanism}, we introduce our privacy mechanism for hiding sensitive genotypes.
In Section~\ref{section:greedy_optimality}, we describe its interpretation as a locally-optimal solution in detail and demonstrate the NP-hardness of finding the optimal ordering in general. The robustness of our privacy mechanism to model mismatch is discussed in Section~\ref{sec:robustness}. In Section~\ref{sec:hmm}, we propose an efficient implementation of the privacy mechanism for hidden Markov models. Simulation experiments are presented in Section~\ref{sec:simulation}. Finally in Section~\ref{sec:conclusion}, we conclude the paper and discuss future directions.

\section{The Genotype-Hiding Problem}
\label{sec:formulation}
Let $\bX=(X_1,\dots,X_n)$ be the user's personal genome sequence of length $n$, and each $X_i$ takes values in the alphabet $\cX$.
The user wishes to share $\bX$ with others, but is concerned about revealing information about certain positions of $\bX$.
To hide the values at these sensitive positions, the user generates a masked version of the data $\bY=(Y_1,\ldots,Y_n)$, which only partially reveals $\bX$.

The desired properties of $\bY$ are given as follows. First, since we expect substitution errors to be considerably more undesirable than erasures in genetic analyses, we impose a constraint that $Y_i$ can be either $X_i$ or the erasure symbol $\ast$. We refer to this property as the \emph{faithfulness condition}, i.e., 
\begin{align}
   Y_i = X_i~\text{or}~\ast.  ~~~~~  \mathbf{(Faithfulness)} 
\end{align}
Note that the alphabet of $Y_i$ is $\mathcal{X}\cup \{\ast\}$.

Next, let $\cK\subset [n] := \{1,\dots,n\}$ be the user-provided set of indices of $\bX$ containing sensitive information. We assume that $\cK$ is chosen irrespective of the sequence (\Ie independently from $\bX$) based on information such as family history or curated disease associations. 
We require that no information about $X_{\cK}=\{X_i:i \in \cK\}$ is revealed when $\bY$ is shared.
In other words, we require that
\begin{align}
\label{eq:privacy}
    I(X_{\cK};\bY) = 0, ~~~~~  \mathbf{(Privacy)}
\end{align}
\noindent where $I(\cdot)$ denotes the mutual information.
We refer to this requirement as the \emph{privacy condition}.
Note that our notion of privacy
is stronger than alternatives such as local differential privacy~\cite{Kairouz14}, which allows a small amount of leakage.
Our work focuses on maximizing the utility over all mechanisms satisfying the perfect privacy condition.

We aim to design a \emph{privacy mechanism} $\wb{\by|\bx}$ to generate $\bY$ from given $\bX$ and $\cK$ such that both the faithfulness and privacy conditions are satisfied.
Here, we consider the ideal scenario where the data generating distribution $\pb{\bx}$ is known to the mechanism.
We discuss the impact of having an inaccurate $\pb{\bx}$ in Section~\ref{sec:robustness}; even under this challenging scenario, we show that the potential privacy leakage is bounded by the divergence between the given $\pb{\bx}$ and the true distribution.  
Note that we use uppercase symbols to represent random variables and lowercase symbols to denote their realizations.

While satisfying the above two conditions, we wish to share as much of $\bX$ as possible.
More precisely, let $e(\bY)$ be the number of erasure symbols in $\bY$.
Our goal is to minimize the expected number of erasures $\mathbb{E}[e(\bY)]$, or equivalently the \emph{erasure rate} $\frac{1}{n} \mathbb{E}[e(\bY)]$, where
\begin{equation}
\mathbb{E}[e(\bY)] 
= \sum_{i=1}^n \mathbb{E}[\mathbbm{1}\{Y_i=\ast\}] 
= \sum_{i=1}^n  \pb{y_i=\ast},
\label{eq:erasure-yi}
\end{equation}
and $\mathbbm{1}\{\cdot\}$ denotes the indicator function.

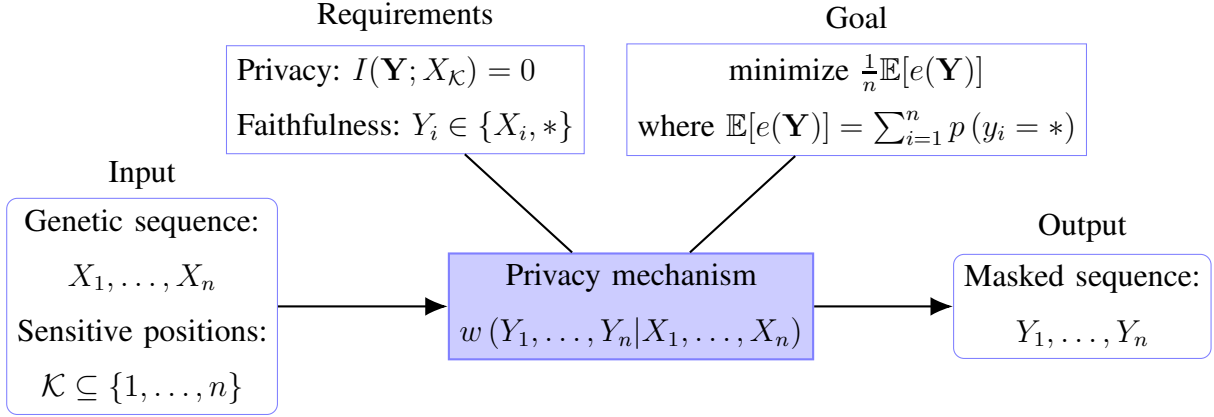
\begin{figure}[t]
\begin{center}
\begin{tikzpicture}
\tikzstyle{rec}=[shape=rectangle,draw=blue!50,fill=blue!20,align=center,thick]
\tikzstyle{condition}=[shape=rectangle,draw=blue!50,align=left]
\tikzstyle{zero}=[shape=rectangle,draw=blue!50,align=center,rounded corners]
\tikzstyle{goal}=[shape=rectangle,draw=blue!50,align=center]

\node[rec] (P) at (9.5,0) {Privacy mechanism \\ $\wb{Y_1,\ldots,Y_n|X_1,\ldots,X_n}$};

\node[zero] (In) at (3,0) {Genetic sequence: \\ $X_1,\ldots,X_n$ \\ Sensitive positions: \\ $\mathcal{K} \subseteq \{1,\ldots,n\}$};

\draw[-{Latex[length=3mm, width=2.5mm]},thick] (In)--(P);

\node[zero] (Out) at (15.5,0) {Masked sequence: \\ $Y_1,\ldots,Y_n$};

\draw[-{Latex[length=3mm, width=2.5mm]},thick] (P)--(Out);

\node[condition] (C) at (6.5,2.75) {Privacy: $I(\bY;X_{\mathcal{K}}) = 0$ \\
Faithfulness: $Y_i \in \{X_i, \ast\}$ };

\draw[thick] (P)--(C);

\node[goal] (G) at (12.5,2.75) {minimize $\frac{1}{n}\mathbb{E}[e(\bY)]$ \\ where $\mathbb{E}[e(\bY)] = \sum_{i=1}^{n} \pb{y_i = \ast}$ };
\draw[thick] (P)--(G);

\node (1) at (6.5,3.85) {Requirements};
\node (4) at (12.5,3.85) {Goal};

\node (2) at (15.5,1.05) {Output};
\node (3) at (3,1.75) {Input};

\end{tikzpicture}
\end{center}
\caption{An illustration of $(n,\mathcal{K})$ genotype-hiding privacy mechanism. 
The mechanism takes as input a genetic sequence along with a set of sensitive positions and outputs a masked sequence with erasures.
We require the faithfulness and privacy conditions to be satisfied, and the goal is to minimize the expected number of erasures in the output.
}
\label{fig:block}
\end{figure}


A formal description of the genotype-hiding problem is given below. We start by defining the privacy mechanism for the genotype-hiding problem as follows.
\begin{definition}
	An $(n,\mathcal{K})$ privacy mechanism for a given data generative distribution $\pb{\bx}$ with input alphabet $\mathcal{X}^n$ and output alphabet $\mathcal{Y}^n$ is defined by a probabilistic encoding function
	\[\mathsf{Enc}: \mathcal{X}^n \rightarrow \mathcal{Y}^n,\]
	where $\mathsf{Enc}$ satisfies both the faithfulness condition ($Y_i\in \{X_i,\ast\},\forall i$) and the privacy condition ($I(X_{\mathcal{K}};\bf{Y}) = 0$). 
\end{definition}

The performance of the privacy mechanism is measured by the expected number of erasures per symbol in an output sequence $\bf{y}$.
This measure captures the distortion between the input and output sequences induced by a set of single-letter erasures.
Following the convention, we define the \emph{rate} of a privacy mechanism as the fraction of positions that are not erased in the output:
\begin{definition}
 	The rate of an $(n,\mathcal{K})$ privacy mechanism for a given data generative distribution $\pb{\bx}$ is defined by $1 - \frac{1}{n} \mathbb{E}[e(\mathsf{Enc}(\bX))]$ per symbol.
\end{definition}

\begin{definition}
    For any given data distribution $\pb{\bx}$, a rate $R$ is achievable if there exists an $(n,\mathcal{K})$  privacy mechanism such that 
    \begin{equation}
      1 - \frac{1}{n} \mathbb{E}[e(\bY)] \geq R,  
    \end{equation}
    where $\bY = \mathsf{Enc}(\bX)$.
\end{definition}
\noindent Clearly, if $R$ is achievable then $R - \epsilon$ for any $\epsilon > 0$ is also achievable by the definition, so we are interested in finding the maximum achievable rate.


It is worth noting that the encoder $\mathsf{Enc}(\cdot)$ can be potentially stochastic, so we may use conditional probabilities $\wb{\by|\bx}$ to represent the encoding function. If we treat conditional probabilities $\wb{\by|\bx}$ where $\bx \in \mathcal{X}^n, \by \in \mathcal{Y}^n$ as decision variables, the genotype-hiding problem can be defined as the following optimization problem:
\begin{equation}
\label{eq:optimization}
\begin{aligned}
\underset{\wb{\by|\bx}}{\text{maximize}}
\quad & 
1 - \frac{1}{n}\sum_{i=1}^n  \pb{y_i = \ast}  &\\
\text{subject to} \quad      & I(X_{\cK};\bY) = 0 & \text{(Privacy)}\\
                         & Y_i  \in \{X_i, \ast\},\forall i &   \text{(Faithfulness)} 
\end{aligned}
\end{equation}
Note that this problem maximizes the information rate (utility) under the stringent privacy constraint such that no information about the sensitive positions is leaked.

If we express the objective and the constraints explicitly in terms of the conditional probabilities $\wb{\by|\bx}$, the optimization problem \eqref{eq:optimization} can be viewed as an instance of linear programming (LP).
However, the scale of the problem is intractable in practice, given the exponential blowup in the number of variables and constraints as the length of the sequence $n$ grows; the number of decision variables is $|\mathcal{X}|^n |\mathcal{Y}|^n$, and the number of constraints is in the order of $|\mathcal{X}|^{|\mathcal{K}|} |\mathcal{Y}|^n  + n\,|\mathcal{X}||\mathcal{Y}|$.

Therefore, the ultimate goal of this paper is to identify a solution to the genotype-hiding problem in a tractable and computationally-efficient manner.
To this end, we first present an achievable privacy mechanism as well as an upper bound on the maximum achievable rate. Then we show that the proposed privacy mechanism is computationally efficient for a particular data generative distribution, namely hidden Markov models, which is of broad interest in our motivating application in genomics.



\section{Performance Bounds} 
\label{sec:bound}

In this section, we state the performance bounds on the achievable rate in the following theorems.
\begin{theorem}
For a given data distribution $\pb{\bx}$, a rate $R$ is achievable if
\begin{equation}
\label{eq:thm-achievability}
R  \leq \frac{1}{n}\sum_{i=1}^n \sum_{x_i \in \cX} \mathbb{E}_{Y_{[i-1]}}\left[ \min_{u \in \cX^{|\cK|}} \pb{x_i|x_{\cK}=u, Y_{[i-1]}}\right] . 
\end{equation}
\end{theorem}
A detailed description of the achievable scheme will be presented in Section~\ref{sec:mechanism}. 

The right-hand side of \eqref{eq:thm-achievability} may appear unconventional, given that conditioning on $Y_{[i-1]}$ for each $i$ makes the probability term generally hard to compute as the sequence length $n$ grows. 
However, this expression corresponds to a sequential mechanism where the encoder generates $Y_1,\ldots,Y_n$ one position at a time, and
an efficient update exists for incrementally expanding the conditioning set.
As an example, in Section~\ref{sec:hmm}, we present a concrete implementation of the privacy mechanism for data distributions governed by hidden Markov models, which indeed allows 
the right-hand side of~\eqref{eq:thm-achievability} to be efficiently computed.

\begin{theorem}
\label{thm:upper}
For a given data distribution $\pb{\bx}$, any achievable rate $R$ must satisfy
\begin{equation}
\label{eq:thm-upper-bound}
R  \leq \frac{1}{n}\sum_{i=1}^n \sum_{x_i \in \cX} \min_{u \in \cX^{|\cK|}} \pb{x_i|x_{\cK}=u}. 
\end{equation}
\end{theorem}
It is worth noting that, given a data distribution $\pb{\bx}$, each summand in the right-hand side of \eqref{eq:thm-upper-bound} represents
the conditional probability of the observation  $x_i$ at coordinate $i$ when the sensitive positions $x_\mathcal{K}$ take on the \emph{least}-likely values, which can be determined from the given $\pb{\bx}$.

\begin{proof}
From \eqref{eq:erasure-yi}, we know that to establish \eqref{eq:thm-upper-bound}, it is sufficient to show 
\begin{equation}
  \pb{y_i \neq \ast} \leq \sum_{x_i \in \cX} \min_{u \in \cX^{|\cK|}} \pb{x_i|x_{\cK}=u}
\end{equation}
for any mechanism satisfying the privacy and faithfulness conditions. Consider
\begin{align}
 \pb{y_i \neq \ast} & = \sum_{y_i \in \cX} \pb{y_i} \nonumber \\
& \utag{a}{=} \sum_{y_i \in \cX} \min_{u} \pb{y_i|x_{\cK}=u} \nonumber\\
& \utag{b}{=} \sum_{y_i \in \cX}  \min_{u} \pb{y_i=x_i|x_{\cK}=u} \nonumber\\
& = \sum_{x_i \in \cX} \min_{u} \pb{x_i|x_{\cK}=u} \pb{y_i=x_i|x_i,x_{\cK}=u} \nonumber \\
& \utag{c}{\leq} \sum_{x_i \in \cX}  \min_{u} \pb{x_i|x_{\cK}=u},
\end{align}
where \uref{a} is due to the fact that $Y_i$ is independent of $X_{\cK}$ (privacy condition); 
\uref{b} follows from the faithfulness condition $Y_i \in \{X_i, \ast\}$; and \uref{c} follows from the fact that probabilities are bounded above by $1$.
\end{proof}


Although not true in general, the upper bounds in \eqref{eq:thm-achievability} and \eqref{eq:thm-upper-bound} match under special circumstances, implying the optimality of an achievable mechanism. 
That is,
\begin{equation}
\label{eq:bound-tight}
	\sum_{x_{i} \in \cX}\sum_{y_{[i-1]}}   \pb{y_{[i-1]}} \min_{x_{\cK}} \pb{x_{i}|x_{\cK},y_{[i-1]}}  =  \sum_{x_i \in \cX} \min_{u \in \cX^{|\cK|}} \pb{x_i|x_{\cK}=u}.
\end{equation}
We observe that a \emph{sufficient condition} for this equality is given by the following:
for any $x_i$, if
\begin{equation}
    u^{\ast} \in \arg\min_{u} \pb{x_i|x_{\cK}=u},
\end{equation}
then
\begin{equation}
  u^{\ast} \in \arg\min_{u} \pb{x_i|x_{\cK}=u,y_{[i-1]}}  
\end{equation}
for all possible $y_{[i-1]}$.
Intuitively, this means that for any given position $x_i$, the least-likely values of the (unobserved) sensitive positions $x_{\mathcal{K}}$
remains the same regardless of the positions that have been previously released in the output $y_{[i-1]}$ during the course of the mechanism.

A special case that satisfies this optimality condition is when random variables $X_1,\ldots,X_n$ form a Markov chain (\Ie $\pb{\bx}$ is induced by a Markov chain), with a single sensitive position. Without loss of generality, we assume $\cK=\{1\}$. 
\begin{corollary}[Markov chain]
\label{corollary}
If $X_1,\ldots,X_n$ forms a Markov chain and the sensitive position is $\mathcal{K}=\{1\}$, then a rate $R$ is achievable if and only if
\begin{equation}
R  \leq \frac{1}{n}\sum_{i=1}^n \sum_{x_i \in \cX} \min_{u \in \cX^{|\cK|}} \pb{x_i|x_{\cK}=u}. 
\end{equation}
\end{corollary}
\noindent It is sufficient to justify the corollary by showing that the aforementioned sufficient condition holds. The proof is included in Appendix~\ref{appendix:corollary}. 

\section{Privacy mechanism}


\label{sec:mechanism}
In this section, we present a privacy mechanism for generating $\bY$ based on a given $\pb{\bx}$, whose 
performance matches the bound given in \eqref{eq:thm-achievability}, while satisfying both faithfulness and privacy conditions. 

Let us first recall the genotype-hiding problem introduced in \eqref{eq:optimization}, i.e.,
\begin{equation}
\begin{aligned}
\underset{\wb{\by|\bx}}{\text{maximize}}
\quad & 
1 - \frac{1}{n}\sum_{i=1}^n  \pb{y_i = \ast}  & \\
\text{subject to} \quad      & I(X_{\cK};\bY) = 0 & \text{(Privacy)}\\
                         & Y_i  \in \{X_i, \ast\},\forall i. &   \text{(Faithfulness)} 
\end{aligned}
\end{equation}
This problem is difficult to solve in its general form given the exponentially growing number of decision variables in $\wb{\by|\bx}$ as the sequence length $n$ grows.
Instead, we adopt a greedy optimization approach, whereby the erasure probability of $y_{i}$ is locally minimized, one position at a time, from $1$ to $n$.
In other words, for each $i=1,\ldots,n$, we solve
\begin{equation}
\label{eq:optimization-sub}
\begin{aligned}
\underset{\wb{y_i|\bx,y_{[i-1]}}}{\text{minimize}}
\quad & 
\pb{y_{i}=\ast|y_{[i-1]}}  & \\
\text{subject to} \quad      & I(X_{\cK};Y_{i}|Y_{[i-1]}) = 0 & \\
                         & Y_i  \in \{X_i, \ast\}, &   
\end{aligned}
\end{equation}
for any given $y_{[i-1]}$. Note that 
\begin{equation}
\label{eq:privacy-chain}
     I(X_{\cK};\bY) = \sum_{i=1}^n I(X_{\cK};Y_{i}|Y_{[i-1]}) = 0,
\end{equation}
by the chain rule, so if the first constraint of \eqref{eq:optimization-sub} is satisfied for all $i$, then the solution preserves the required privacy constraint $I(X_{\cK};\bY) = 0$ as defined in \eqref{eq:privacy}. The second constraint is inherited directly from the faithfulness condition. In other words, any solution satisfying the constraints of \eqref{eq:optimization-sub} for all $i$ will naturally be a feasible solution to the genotype-hiding problem in~\eqref{eq:optimization}.

We observe that solving the local optimization problem \eqref{eq:optimization-sub} gives rise to a sequential mechanism for generating $\bY$.
That is, we generate $\bY$ one position at a time, where the conditional distribution for $Y_{i}$ may depend on the values of $Y_{1},\ldots,Y_{i-1}$ that have been previously generated. 
The following defines our chosen privacy mechanism for any given position $i$, which is in fact an optimal solution to the local optimization problem \eqref{eq:optimization-sub}.
A detailed proof of the local optimality of this scheme is deferred to Section~\ref{section:greedy_optimality}.





\noindent\textbf{Privacy mechanism:} 
Generate each $Y_{i}$ according to the following conditional distribution
\begin{equation}
\label{eq:encode}
\wb{y_{i}|x_{i},x_{\cK},y_{[i-1]}}  = 
\begin{cases}
    \frac{\min_{u \in \cX^{|\cK|}} \pb{x_{i}|x_{\cK}=u,y_{[i-1]}}}{\pb{x_{i}|x_{\cK},y_{[i-1]}}}, & \text{if }y_{i}=x_{i},\\
    1 - \frac{\min_{u \in \cX^{|\cK|}} \pb{x_{i}|x_{\cK}=u,y_{[i-1]}}}{\pb{x_i|x_{\cK},y_{[i-1]}}}, & \text{if }y_{i}=\ast, \\
    0, & \text{otherwise,}
\end{cases}
\end{equation}
for any $x_{i}$, $x_{\cK}$ and $y_{[i-1]}$, where $[i-1]:=\{1,\ldots,i-1\}$. 

The expression for the erasure probability in the above mechanism can be intuitively understood as follows.
We first identify 
the values of the sensitive positions with the smallest likelihood of generating the observed symbol $x_i$ at the $i$-th position (as indicated by the numerator in the fractional term), conditioned on the previously released positions $y_{[i-1]}$.
Note that $u$ is an auxiliary variable denoting the possible values in the alphabet $\cX^{|\cK|}$, whereas $x_\cK$ denotes the observed values at the sensitive positions.
We then choose the erasure probability such that, the probability of releasing the original symbol (without erasure) becomes identical among different hypothetical values of $x_\cK$, thus ensuring privacy.


It is worth noting that our privacy mechanism satisfies the faithfulness condition (\Ie $y_i \in \{x_i,\ast\}$) by design, so we only need to verify that it satisfies the privacy constraint \eqref{eq:privacy}. Before verifying the privacy constraint, we note the following properties of the mechanism.
\begin{enumerate}[label=(\arabic*)]
    \item If $i \in \cK$, then
    \begin{equation}
        \min_{u \in \cX^{|\cK|}} \pb{x_{i}|x_{\cK}=u,y_{[i-1]}} = 0,
    \end{equation}
which yields
\begin{equation}
\wb{y_{i}=\ast|x_{i},x_{\cK},y_{[i-1]}}
    = 1.
\end{equation}
This implies that $X_{i}$ is \emph{always erased} if it corresponds to one of the sensitive positions in $\cK$.
    
\item We notice from \eqref{eq:encode} that 
$X_{i}$ is \emph{not} erased with some nonzero probability, so this mechanism is strictly better than the na\"{i}ve approach of always erasing any position that have a nonzero correlation with the sensitive positions.
\end{enumerate}

\begin{proof}[Proof of privacy]
To show that the proposed mechanism in \eqref{eq:encode} satisfies the privacy condition \eqref{eq:privacy}, 
it is sufficient to show 
\begin{equation}
\label{eq:privacy-conditional}
    I(Y_{i};X_{\cK}|Y_{1},\ldots,Y_{i-1}) = 0,
\end{equation}
for all $i=1,\ldots,n$, since this implies
\begin{equation}
     I(X_{\cK};\bY) = \sum_{i=1}^n I(X_{\cK};Y_{i}|Y_{[i-1]}) = 0
\end{equation}
by the chain rule.
To establish \eqref{eq:privacy-conditional}, we will equivalently prove that 
\begin{equation}
\label{eq:privacy-conditional-prob}
    \pb{y_{i}|x_{\cK},y_{[i-1]}} = \pb{y_{i}|y_{[i-1]}}
\end{equation}
for any $x_{\cK}$, $y_{[i-1]}$ and $y_i$.
Since 
\begin{equation}
    \pb{y_{i}|x_{\cK},y_{[i-1]}} 
     = \sum_{x_{i} \in \cX} \pb{x_{i}|x_{\cK},y_{[i-1]}} \wb{y_{i}|x_{i},x_{\cK},y_{[i-1]}},
\end{equation}
by substituting \eqref{eq:encode}, we have 
\begin{align}
  \pb{y_{i}=\ast|x_{\cK},y_{[i-1]}} 
  &  = \sum_{x_{i}} \pb{x_{i}|x_{\cK},y_{[i-1]}} \wb{y_{i}=\ast|x_{i},x_{\cK},y_{[i-1]}} \nonumber\\
  &  = 1 - \sum_{x_{i} \in \cX}\min_{u \in \cX^{|\cK|}}\pb{x_{i}|x_{\cK}=u,y_{[i-1]}}. 
  \label{eq:privacy-verify-1}
\end{align}
Similarly, for $y_i \in \cX$, we have 
\begin{align}
 \pb{y_{i}|x_{\cK},y_{[i-1]}} 
&  = \sum_{x_{i} \in \cX} \pb{x_{i}|x_{\cK},y_{[i-1]}} \wb{y_{i}=x_{i}|x_{i},x_{\cK},y_{[i-1]}} \nonumber\\
&  = \sum_{x_{i} \in \cX}\min_{u \in \cX^{|\cK|}}\pb{x_{i}|x_{\cK}=u,y_{[i-1]}}. \label{eq:privacy-verify-2}
\end{align}
We can observe that the right-hand sides of both \eqref{eq:privacy-verify-1} and \eqref{eq:privacy-verify-2} are independent of $x_{\cK}$, and hence by combining \eqref{eq:privacy-verify-1} and \eqref{eq:privacy-verify-2}, we have
\begin{equation}
\label{eq:revision-t1}
    \pb{y_{i}|x_{\cK},y_{[i-1]}} = \pb{y_{i}|y_{[i-1]}},
\end{equation}
for any $x_{\cK}$, $y_{[i-1]}$ and $y_i$, which finishes the proof of \eqref{eq:privacy-conditional-prob}. 
\end{proof}

Finally, we can easily verify that our sequential privacy mechanism \eqref{eq:encode} achieves the rate
\begin{align}
    1 - \frac{1}{n} \sum_{i=1}^{n} \pb{y_i = \ast} 
    & = 1 - \frac{1}{n} \sum_{i=1}^{n} \sum_{y_{[i-1]}}  \pb{y_{i}=\ast|y_{[i-1]}} \pb{y_{[i-1]}} \nonumber \\
    & \utag{a}{=} 1 - \frac{1}{n} \sum_{i=1}^{n} \sum_{y_{[i-1]}}   \pb{y_{[i-1]}} \left( 1 - \sum_{x_{i} \in \cX}\min_{u \in \cX^{|\cK|}}\pb{x_{i}|x_{\cK}=u,y_{[i-1]}} \right) \nonumber \\
    & =  \frac{1}{n} \sum_{i=1}^{n} \sum_{y_{[i-1]}}   \pb{y_{[i-1]}}  \sum_{x_{i} \in \cX}\min_{u \in \cX^{|\cK|}}\pb{x_{i}|x_{\cK}=u,y_{[i-1]}}  \nonumber \\
    & =  \frac{1}{n} \sum_{i=1}^{n} \sum_{x_{i} \in \cX}\sum_{y_{[i-1]}}   \pb{y_{[i-1]}} \min_{x_{\cK}} \pb{x_{i}|x_{\cK},y_{[i-1]}}, 
\end{align}
where \uref{a} follows by \eqref{eq:privacy-verify-1} and \eqref{eq:revision-t1}. The final expression is identical to the right-hand side of \eqref{eq:thm-achievability} as desired.


\vspace{0.5em}
\noindent \emph{Example.}
We present an example to illustrate the operations of the proposed privacy mechanism in a simplified setting.
Let us consider a data distribution $\pb{\bx}$ where $X_1,\ldots,X_n$ form a Markov chain, as in Corollary~\ref{corollary}, and a single sensitive position $\mathcal{K}= \{1\}$.

By inspecting the privacy mechanism in \eqref{eq:encode}, we know that if $y_{i-1} \neq \ast$ for some $i>1$, then 
\begin{align}
\pb{x_i|x_{\mathcal{K}} = u, y_{[i-1]}}
& = \pb{x_i|x_{\mathcal{K}} = u, y_{[i-1]}, x_{i-1}=y_{i-1}} 
 = \pb{x_i|x_{i-1}=y_{i-1}}, 
\end{align}
for any $x_i$ and $y_{[i-1]}$ by the Markov property and the fact that $\mathcal{K} = \{1\}$. This implies that
\begin{align}
\wb{y_{i} = x_i|x_{i},x_{\cK},y_{[i-1]}}  = \frac{\min_{u \in \cX^{|\cK|}} \pb{x_{i}|x_{\cK}=u,y_{[i-1]}}}{\pb{x_{i}|x_{\cK},y_{[i-1]}}} 
& = \frac{\pb{x_i|x_{i-1}=y_{i-1}}}{\pb{x_i|x_{i-1}=y_{i-1}}}  = 1,
\end{align}
which means that if $y_{i-1} \neq \ast$ then $y_{i} \neq \ast$ with probability one.


Thus, when $\pb{\bx}$ is specified by a Markov chain, we see that the privacy mechanism erases all positions within a window from the sensitive position and releases the rest without erasure, and the size of the window is stochastically chosen.
This observation suggests that, in contrast to the heuristic approach of deterministically choosing a window for erasure, our mechanism introduces additional uncertainty about sensitive data (in fact achieving perfect privacy) by randomizing the choice of the window.
Later in Section~\ref{sec:simulation}, we present a simulation experiment comparing our mechanism with the deterministic window-based erasure approach with respect to the privacy-utility trade-off, based on a more realistic data distribution defined by hidden Markov models.

\section{
Local optimality
}
\label{section:greedy_optimality}


In the previous section, we proposed a privacy mechanism for the genotype-hiding problem satisfying both privacy and faithfulness conditions.
Here, we provide further insights into the optimality of the proposed mechanism.
We first prove that the mechanism is indeed an optimal solution to the local optimization problem in \eqref{eq:optimization-sub} as claimed, and thus can be viewed as a greedy solution to the general genotype-hiding problem in \eqref{eq:optimization} given a fixed variable ordering (\Ie the order in which $Y_i$'s are sampled).
We then present a negative result to inform future investigation, showing that finding an optimal variable ordering for the mechanism is intractable (NP-hard) in general, thus illustrating the limits of current techniques in achieving global optimality.

\subsection{
Optimality with respect to the local optimization problem} 
\label{subsec:interpretation}
Let us first recall the local optimization problem \eqref{eq:optimization-sub}, \Ie 
\begin{equation}
\begin{aligned}
\underset{\wb{y_i|\bx,y_{[i-1]}}}{\text{minimize}}
\quad & 
\pb{y_{i}=\ast|y_{[i-1]}}  & \\
\text{subject to} \quad      & I(X_{\cK};Y_{i}|Y_{[i-1]}) = 0 & \\
                         & Y_i  \in \{X_i, \ast\}. &   
\end{aligned}
\end{equation}
As we have shown,
\begin{equation}
     I(X_{\cK};\bY) = \sum_{i=1}^n I(X_{\cK};Y_{i}|Y_{[i-1]}) = 0,
\end{equation}
by the chain rule, so any solution satisfying the constraints of \eqref{eq:optimization-sub} for all $i$ is a feasible solution to the general genotype-hiding problem in~\eqref{eq:optimization}. 

We now show that the privacy mechanism in \eqref{eq:encode} is optimal with respect to the above optimization problem.
First, for any given $y_{[i-1]}$, note that  
\begin{align}
     \pb{y_{i}=\ast|y_{[i-1]}} 
    &  = 1 - \sum_{y_{i} \in \cX}\pb{y_{i}|y_{[i-1]}} \nonumber\\
    &  \utag{a}{=} 1 - \sum_{y_{i} \in \cX}\min_{x_{\cK}} \pb{y_{i}|x_{\cK},y_{[i-1]}} \nonumber\\
    &  \utag{b}{=} 1 - \sum_{y_{i} \in \cX}\min_{x_{\cK}} \pb{y_{i}=x_{i}|x_{\cK},y_{[i-1]}} \nonumber\\
    &  = 1 - \sum_{x_{i} \in \cX}\min_{x_{\cK}} \pb{x_{i}|x_{\cK},y_{[i-1]}}\wb{y_{i}=x_{i}|x_{i},x_{\cK},y_{[i-1]}} \nonumber\\
    &  \utag{c}{\geq} 1 - \sum_{x_{i} \in \cX}\min_{x_{\cK}} \pb{x_{i}|x_{\cK},y_{[i-1]}},
\end{align}
where \uref{a} follows from the privacy condition, \uref{b} follows from the faithfulness condition, and \uref{c} holds because probability values are at most $1$. 
This implies that any feasible solution to the local optimization problem \eqref{eq:optimization-sub} has to satisfy 
\begin{equation}
     \pb{y_{i}=\ast|y_{[i-1]}}  \geq 1 - \sum_{x_{i} \in \cX}\min_{x_{\cK}} \pb{x_{i}|x_{\cK},y_{[i-1]}},
\end{equation}
and that it is optimal if the last step holds with equality, \Ie 
\begin{equation}
    \min_{x_{\cK}} \pb{x_{i}|x_{\cK},y_{[i-1]}}\wb{y_{i}=x_{i}|x_{i},x_{\cK},y_{[i-1]}} =\min_{x_{\cK}} \pb{x_{i}|x_{\cK},y_{[i-1]}},
\end{equation}
for any $x_i$ and $y_{[i-1]}$.

By plugging in the proposed mechanism in \eqref{eq:encode}, we have
\begin{align}
    &\min_{x_{\cK}} \pb{x_{i}|x_{\cK},y_{[i-1]}}\wb{y_{i}=x_{i}|x_{i},x_{\cK},y_{[i-1]}} \nonumber \\
    = &\min_{x_{\cK}} \pb{x_{i}|x_{\cK},y_{[i-1]}}\frac{\min_{u \in \cX^{|\cK|}} \pb{x_{i}|x_{\cK}=u,y_{[i-1]}}}{\pb{x_{i}|x_{\cK},y_{[i-1]}}} \nonumber\\
    = &\min_{x_{\cK}} \min_{u \in \cX^{|\cK|}} \pb{x_{i}|x_{\cK}=u,y_{[i-1]}} \nonumber \\
    = &\min_{x_{\cK}} \pb{x_{i}|x_{\cK},y_{[i-1]}},
\end{align}
where the last step follows because the two minimizations, both over the alphabet of $X_{\mathcal{K}}$, are equivalent and can be merged. 
This implies that the mechanism \eqref{eq:encode} attains the minimum probability of erasing $Y_{i}$ and thus is an optimal solution to the local optimization problem \eqref{eq:optimization-sub}. 
Therefore, our sequential privacy mechanism can be viewed as a locally-optimal algorithm for solving the general genotype-hiding problem \eqref{eq:optimization}, given a fixed variable ordering. 


\subsection{NP-hardness of finding an optimal variable ordering}
\label{sec:NPhard}
So far, we considered the privacy mechanism that generates a masked sequence $Y_1,\dots,Y_n$ in a linear order from $1$ to $n$.
A natural question is then whether this linear ordering is optimal in terms of the erasure rate that the locally-optimal mechanism achieves.
Here, we illustrate the difficulty of determining the optimal variable ordering for the mechanism from a complexity theory perspective, by proving that it is NP-hard in general.
This suggests that devising an efficient mechanism with better optimality guarantees in the general setting requires additional assumptions or techniques to circumvent this impossibility result, which is an interesting direction for further research.

\revise{

}

To formalize the problem, let $(o_1,\ldots,o_n)$ be any permutation of $(1,\ldots,n)$.
We consider generating $\bY$ in the order of $o_1,\ldots,o_n$ instead.
In this setting, the privacy mechanism \eqref{eq:encode} is defined by the conditional distribution
\begin{equation}
\label{eq:encode-modified}
\wb{y_{o_i}|x_{o_i},x_{\cK},y_{o_{[i-1]}}} = 
\begin{cases}
    \frac{\min_{u \in \cX^{|\cK|}} \pb{x_{o_i}|x_{\cK}=u,y_{o_{[i-1]}}}}{\pb{x_{o_i}|x_{\cK},y_{o_{[i-1]}}}}, & \text{if }y_{o_i}=x_{o_i},\\
    1 - \frac{\min_{u \in \cX^{|\cK|}} \pb{x_{o_i}|x_{\cK}=u,y_{o_{[i-1]}}}}{\pb{x_{o_i}|x_{\cK},y_{o_{[i-1]}}}}, & \text{if }y_{o_i}=\ast,\\
    0, & \text{otherwise},
\end{cases}
\end{equation}
for any $x_{o_i}$, $x_{\cK}$ and $y_{o_{[i-1]}}$, where $o_{[i-1]}:=\{o_1,\ldots,o_{i-1}\}$.
It is easy to see that the faithfulness and privacy conditions are still satisfied regardless of the ordering. 

In the following, we show that finding the best ordering $(o_1,\ldots,o_n)$ that minimizes the erasure rate of the mechanism is NP-hard by constructing a polynomial-time reduction of the well-known \emph{hitting set} problem~\cite{cormen2009introduction} to our problem.
More specifically, given an arbitrary instance of a hitting set problem, we construct an instance of the genotype-hiding problem for which finding the optimal ordering for the privacy mechanism is equivalent to solving the original hitting set problem.

At the core of this reduction is a bipartite graph, illustrated in Figure~\ref{fig:NP}, which we use to represent both an instance of the hitting set problem and to construct a corresponding instance of the genotype-hiding problem, as we explain in detail below.
To clarify the dimensions of the problems upfront, note that we represent a hitting set problem for $k$ sets over $m$ elements using a bipartite graph with $m$ left nodes and $k$ right nodes, and the resulting genotype-hiding problem is over a sequence of length $n=m+k$ with $k$ sensitive positions ($|\cK|=k$) and a specially constructed $\pb{\bx}$. 

\begin{figure}[htbp]
\begin{center}
\begin{tikzpicture}
\tikzstyle{fnode}=[shape=rectangle,draw=blue!50,fill=blue!20]
\tikzstyle{vnode}=[shape=circle,draw=blue!50,fill=blue!20]
\node[vnode] (v1) at (0,0) [label=left:$v_1$]{};
\node[vnode] (v2) at (0,-0.75) [label=left:$v_2$]{};
\node[vnode] (v3) at (0,-2) [label=left:$v_i$]{};
\node[vnode] (v4) at (0,-2.75) [label=left:$v_{i+1}$]{};
\node[vnode] (v5) at (0,-4) [label=left:$v_m$]{};
\node at (0,-1.25) {$\vdots$};
\node at (0,-3.25) {$\vdots$};
\node[fnode] (s1) at (4,-0.5) [label=right:$S_1$]{};
\node[fnode] (s2) at (4,-1.25) [label=right:$S_2$]{};
\node[fnode] (s3) at (4,-2.5) [label=right:$S_j$]{};
\node[fnode] (s4) at (4,-3.75) [label=right:$S_{k}$]{};
\node at (4,-1.75) {$\vdots$};
\node at (4,-3) {$\vdots$};
\draw (v1) -- (s1); 
\draw (v1) -- (s2);
\draw (s2) -- (v2);
\draw (v3) -- (s3);
\node at (2,-2) {$b_{i,j}$};
\draw (s3) -- (v4);
\draw (v4) -- (s4);
\draw (s3) -- (v1);
\draw (v5) -- (s4);

\draw [dashed] (-1.25,0.5) rectangle (0.75,-4.5);
\node at (-0.25,1) {$U=\{v_1,\ldots,v_m\}$};

\draw [dashed] (3.25,0.5) rectangle (5,-4.5);
\node[align=center] at (4.125,1) { $S_1,\dots,S_k\subseteq U$};
\end{tikzpicture}
\end{center}
\caption{
A graphical illustration of the bipartite graph used in our NP-hardness proof, representing an instance of the hitting set problem. The universe $U$ is represented by vertices on the left, sets are represented by vertices on the right, and the edges represent the inclusion of elements in each set. To facilitate reduction to the genotype-hiding problem, we associate each edge with an independent and uniformly random bit $b_{i,j}$. 
}
\label{fig:NP}
\end{figure}
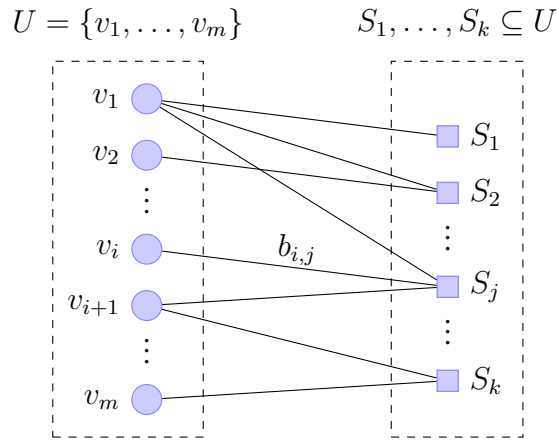

We first 
review the hitting set problem.
Consider a universe $U=\{v_1,\ldots,v_m\}$ and a collection of non-empty subsets $\mathcal{S} = \{S_1,\ldots,S_k\}$ such that $S_j \subseteq U$ for all $j \in [k]$. Without loss of generality, assume that $U = \bigcup_{j=1}^k  S_j$, and  $U = [m]$. A universe $U$ and sets $\{S_1,\ldots,S_k\}$ can be represented by a bipartite graph, as depicted in Fig.~\ref{fig:NP}.
The goal of the hitting set problem is to find the minimum cardinality $h^\ast$ of a set $V\subseteq U$ that satisfies $V\cap S_i\neq \emptyset$ for all $i$, that is
\begin{equation}
\label{eq:hitting-set}
	h^{\ast} = \min_{V \subseteq U: V \cap S_j \neq \emptyset, \forall j \in [k]} |V|.
\end{equation}

Next, we construct the corresponding genotype-hiding problem from the given hitting set problem instance $(U,\mathcal{S})$.
For any $i \in [m]$,  $j \in [k]$ such that $i \in S_j$, let $b_{i,j}$ be a random variable which is independently and uniformly drawn from $\{0,1\}$. In other words, each edge in the bipartite graph is associated with a random bit $b_{i,j}$ (see Fig.~\ref{fig:NP}).
Then, we define $\bX$ to be a sequence of length $n=m+k$ as follows.
Let $X_i$ for $i \in [m]$ be a tuple of random bits associated with edges connected to node $v_i$, \Ie
\begin{equation}
 X_i=(b_{i,j_1},\dots,b_{i,j_r}),   
\end{equation}
where $\{j_1,\dots,j_r\}=\{j:i\in S_j\}$. Next, let $X_{m+j}$ for $j\in [k]$ be
\begin{equation}
    X_{m+j} = \bigoplus_{i \in S_j} b_{i,j},
\end{equation}
which can be viewed as a parity check bit over the edges connected to node $S_j$.
In other words, the first $m$ positions of the sequence are uniform and independently distributed symbols (a tuple of random bits), whereas the remaining $k$ positions are parity check bits defined over the first $m$ positions.

Note that the joint distribution $\pb{\bf{x}}=\pb{x_1,\ldots,x_{m+k}}$ is succinctly characterized by the random bits $b_{i,j}$'s and the associated bipartite graph, and thus the description of the genotype-hiding problem can be generated in polynomial time with respect to $m$ and $k$.
In the following, we refer to the above data generating distribution as $\pb{\bx;U,\mathcal{S}}$, with respect to which the corresponding genotype-hiding problem is defined.

\begin{theorem}
\label{theorem:np}
Given a data generating distribution $p\left(\bx;U,S\right)$ for a sequence of length $n=m+k$ and sensitive positions $\cK=\{m+1,\dots,m+k\}$, finding the best ordering $(o_1,\ldots,o_{m+k})$ that minimizes the erasure rate of our mechanism \eqref{eq:encode-modified} is NP-hard.
\end{theorem}

We provide a sketch of the proof here and defer the details to the Appendix.
First, we note the key property of $p\left(\bx;U,S\right)$ that whether or not our mechanism erases the $o_i$-th position is deterministic given the variable ordering, as stated in the following lemma.
\begin{lemma}
\label{lemma:deterministic}
Given a data generating distribution $p\left(\bx;U,S\right)$ 
for a sequence of length $n=m+k$ and sensitive positions $\cK=\{m+1,\dots,m+k\}$,
the conditional sampling distribution of our privacy mechanism satisfies
\begin{equation}
\label{eq:NP-determine}
\wb{y_{o_i}=\ast|x_{o_i},x_{\cK},y_{o_{[i-1]}}} \in \{0,1\}
\end{equation} 
for all $i$, given any ordering $\pi=(o_1,\ldots,o_{m+k})$. 
\end{lemma}
\begin{proof}
See Appendix~\ref{appendix:lemma-deterministic}.
\end{proof}

As a result of Lemma~\ref{lemma:deterministic}, the overall erasure rate of the privacy mechanism can be calculated simply by counting the number of erased positions.
Note that, if $o_i \in \cK$, then
\begin{equation}
\wb{y_{o_i}=\ast|x_{o_i},x_{\cK},y_{o_{[i-1]}}}
    = 1,
\end{equation}
regardless of the ordering as we have previously shown.
Thus, we need to compare only the erased indices in $[m]=[m+k]\backslash \cK$ for finding the best ordering. 

Let $E_{\pi}$ be the set of erased indices in $[m]$ for a given ordering $\pi=(o_1,\ldots,o_n)$, \Ie 
\begin{equation}
	E_{\pi} = \left\{i: y_i = \ast, i \in [m] \right\},
\end{equation} 
where the distribution over $\bY$ is determined by the privacy mechanism.
Then, finding the best ordering corresponds to finding $\pi$ that leads to the minimum cardinality $e^\ast$ of the corresponding $E_\pi$:
\begin{equation}
\label{eq:optimal-order}
	e^{\ast} = \min_{\pi} |E_{\pi}|.
\end{equation}

Intuitively, whether a particular index $i\in [m]$ is included in $E_\pi$ can be easily determined based on the bipartite graph representation of the underlying hitting set problem (see Fig.~\ref{fig:NP}) as follows.
The ordering $\pi=(o_1,\ldots,o_{m+k})$ specifies the order in which the $m$ nodes on the left-hand side of the graph, each with a corresponding $X_i$, is visited by the mechanism (disregarding the sensitive indices  $o_i \notin [m]$, which are always erased).
As we show in the proof of Lemma~\ref{lemma:deterministic}, when we visit the node $o_i \in [m]$, $X_{o_i}$ is erased if and only if there exists a node $j\in [k]$ on the right-hand side of the graph that is connected to $o_i$ and only to other nodes (if any) that are previously visited \emph{and} not erased.
The presence of such a node $j$ indicates that the sensitive variable $X_{m+j}$ is directly revealed by $X_{o_i}$ (since the rest of random bits contributing to $X_{m+j}$ are already released in $\bY$ without erasure), while the absence of such $j$ indicates the existence of other positions that are erased or have not been released that fully mask the correlation between $X_{o_i}$ and the sensitive positions.

Finally, we complete the reduction by showing that solving \eqref{eq:optimal-order} also produces a solution for the hitting set problem \eqref{eq:hitting-set}, \Ie $e^{\ast} = h^{\ast}$.
This is achieved by showing both that the set of erased indices $E_{\pi}$ is in fact a valid hitting set ($e^{\ast} \geq h^{\ast}$),  and that there exists an ordering $\pi$ satisfying $|E_{\pi}| \leq |V|$ for any given hitting set $V$ ($e^{\ast} \leq h^{\ast}$).
A detailed proof is included in Appendix~\ref{appendix:hitting-set}.

Since the hitting set problem is equivalent to the set cover problem and is well-known to be NP-hard,
our reduction proves that finding the best ordering $\pi$ for our privacy mechanism given any $\pb{\bx}$ and $\cK$ is also NP-hard. 
We note that this result does not preclude the possibility that for a restricted class of genotype-hiding problems (\Eg with a structured $\pb{\bx}$ defined by HMMs), one could still find an efficient polynomial-time algorithm for determining the optimal variable ordering, which remains an interesting open question.

\section{Robustness}
\label{sec:robustness}
In this section, we discuss the robustness of our mechanism with respect to the underlying data distribution. 
In our formulation of the privacy mechanism, the distribution (or the data generative model) $\pb{\bf{x}}$, from which the input genome sequence originated, is assumed to be known.
In practice, one can only empirically estimate this distribution based on existing data resources, \Eg by obtaining maximum likelihood estimates of the model parameters based on a large collection of reference genomes in public data repositories.
Consequently, the generative model used by the mechanism is bound to have deviations from the true generative process, both in terms of the limitations of the model as well as the noisy estimation of the parameters. 
These discrepancies can potentially lead to privacy leakage if the adversary has access to a more accurate distribution for the underlying input. 
Here, we study the potential privacy leakage under the worst-case scenario, where the adversary has access to the true underlying distribution.
We bound the potential leakage as a function of the distance between the data distribution used by the mechanism and the true underlying distribution, suggesting that our mechanism is robust to small deviations in the noisy data distribution we expect to encounter in real-world use cases.


We denote the noisy data distribution used by the mechanism by $\qb{\bf{x}}$ and the true distribution by $\pb{\bf{x}}$.
The privacy mechanism constructs the sampling distribution $w(\bf{y}|\bf{x})$ based on the available $\qb{\bf{x}}$ such that the output $\bY$ is independent of sensitive genotypes $X_{\cK}$ with respect to the joint distribution $\qb{\bx,\by}$ induced by $q(\bx)$ and the mechanism $w(\bf{y}|\bf{x})$, \Ie
\begin{equation}
\label{eq:mismatch-private}
\qb{x_{\cK},\bf{y}}=
 	\sum_{x_{[n]\backslash\cK}} \qb{\bf{x},\bf{y}}= \sum_{x_{[n]\backslash\cK}} w(\mathbf{y}|\mathbf{x}) \qb{\bf{x}} = \qb{x_{\cK}} \qb{\bf{y}}.
 \end{equation} 
Since $\bX$ is actually generated from $\pb{\bx}$ not $\qb{\bx}$, we also define the true joint distribution $\pb{\bx,\by}$ induced by $\pb{\bx}$ and the mechanism $w(\bf{y}|\bf{x})$; note that the mechanism is still based on $\qb{\bx}$.


Then, we can measure the unforeseen privacy leakage due to the mismatch in data distribution by the mutual information $I(\pb{x_{\cK}};\pb{\bf{y}})$ between the sensitive genotypes and the output sequence with respect to $\pb{\bx,\by}$, as follows:
\begin{align}
	 I(\pb{x_{\cK}};\pb{\bf{y}}) 
	&  = \sum_{x_{\cK},\bf{y}} \pb{x_{\cK},\bf{y}} \log \frac{\pb{x_{\cK},\bf{y}}}{\pb{\bf{y}}\pb{x_{\cK}}} \nonumber \\
	& = \sum_{x_{\cK},\bf{y}} \pb{x_{\cK},\bf{y}} \log \frac{\pb{x_{\cK},\bf{y}}\qb{x_{\cK},\bf{y}}}{\pb{\bf{y}}\pb{x_{\cK}}\qb{x_{\cK},\bf{y}}} \nonumber \\
	&  \utag{a}{=} \sum_{x_{\cK},\bf{y}} \pb{x_{\cK},\bf{y}} \log \frac{\pb{x_{\cK},\bf{y}}\qb{x_{\cK}}\qb{\bf{y}}}{\pb{\bf{y}}\pb{x_{\cK}}\qb{x_{\cK},\bf{y}}} \nonumber \\
	&  = D(\pb{x_{\cK},\bf{y}}||\qb{x_{\cK},\bf{y}}) - D(\pb{x_{\cK}}||\qb{x_{\cK}}) - D(\pb{\bf{y}}||\qb{\bf{y}}), \label{eq:robust-intermediate}
\end{align}
where $D(\cdot||\cdot)$ denotes relative entropy or equivalently Kullback-Leibler (KL) divergence, and \uref{a} follows from \eqref{eq:mismatch-private}.
This leads to the following theorem.
\begin{theorem}
\label{thm:leakage}
$I(\pb{x_{\cK}};\pb{\bf{y}})  \leq D(\pb{\bf{x}}||\qb{\bf{x}})$.
\end{theorem}
\begin{proof}
See Appendix~\ref{appendix:leakage}.
\end{proof}

This result implies that the amount of privacy leakage due to the potential mismatch between the data distribution used by the mechanism and the true underlying generative process gracefully scales with the extent to which the two distributions diverge.


\section{Privacy mechanism for Hidden Markov models}
\label{sec:hmm}

Thus far, we considered the data generative model $\pb{\bx}$ of the privacy mechanism to be an arbitrary distribution.
Here, we address a particular form of $\pb{\bx}$ of great interest for our application setting in genomics, namely the Li and Stephens model~\cite{Li03}, which is based on a hidden Markov model.
This model is widely adopted in genetics for a wide range of tasks that require a probabilistic model of the genome~\cite{song2016li}.
For this class of $\pb{\bx}$, we propose an efficient algorithm to implement the privacy mechanism introduced in Section~\ref{sec:mechanism}.

\subsection{Review of hidden Markov models for genomes}
\label{sec:HMM}
The classical hidden Markov model (HMM) describing the distribution of personal genomes \cite{Li03} is as follows. First, let $\bX=(X_1,X_2,\ldots,X_n)$ represent an individual's (haplotype) genetic sequence of length $n$.
Following standard practice in genetics, we adopt a binary alphabet $\mathcal{X}=\{0,1\}$ for each element $X_i$, representing whether the observed nucleotide is identical to the one in the reference human genome (called \emph{reference allele}) or not (\emph{alternative allele}).   
In addition, we are given a reference dataset of $m$ personal genome sequences $\cH=\{{\bf h}_j: j=1,\ldots,m\}$, where each sequence ${\bf h}_j$ is of length $n$. The $i$-th coordinate of ${\bf h}_j$ is denoted by $h_{i,j}$, which also takes a value in $\mathcal{X}$.

In this model, $\bX$ is viewed as a ``mosaic'' of reference sequences in $\cH$ with potential substitution errors arising from mutations or experimental noise in sequencing. Formally, $\bX$ depends on a sequence of hidden states $\{S_i\}_{i=1}^n$ forming a Markov chain, where each $S_i$ takes an integer in the range $\{1,\ldots,m\}$, representing an index into $\mathcal{H}$. 
Without loss of generality, we assume that the initial state $S_1$ is uniformly distributed over $\{1,\ldots,m\}$.
The transition probability $\pi_{i,j}$ from state $i$ to $j$ is set to $\frac{\epsilon}{m-1}$ and $1-  \epsilon$ for $i \neq j$ and $i = j$, respectively.
The parameter $\epsilon$ is often called the recombination probability; in the following we also use the term \emph{crossover probability} to refer to this quantity.

Next, each $X_i$ is sampled based on the hidden state $S_i$ by copying the corresponding symbol in the selected reference sequence with a small probability of error. In other words, $X_i$ is equal to the symbol in the $i$-th position of $\mathbf{h}_{S_i}$ with \emph{error probability} $\theta$. 
The overall data distribution $\pb{\bx}$ is fully specified by the tuple $(\cH,\epsilon,\theta)$.
We provide a graphical illustration of $\pb{\bx}$ in Fig.~\ref{fig:model}.
In our work, we treat the parameters of the above model as given. In practice, these parameters are estimated from a large collection of reference genomes, \Eg including hundreds of thousands of individuals, which are available in public data repositories such as the UK Biobank~\cite{Sudlow15}. 

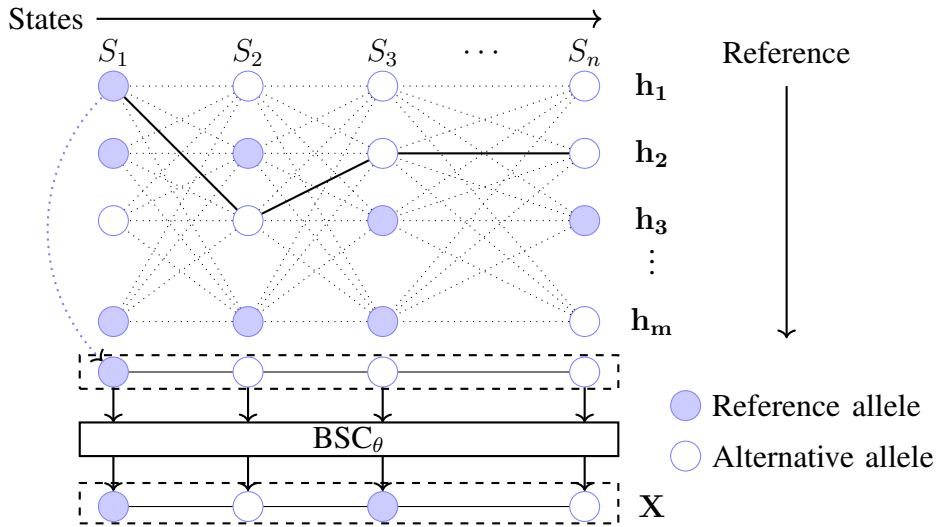
\begin{figure*}[t!]
\begin{center}
\begin{tikzpicture}[scale=0.895]
\tikzstyle{zero}=[shape=circle,draw=blue!50,fill=blue!20]
\tikzstyle{one}=[shape=circle,draw=blue!50]
\tikzstyle{observationzero}=[shape=rectangle,draw=blue!50,fill=blue!20]
\tikzstyle{observationone}=[shape=rectangle,draw=blue!50]
\tikzstyle{lightedge}=[dotted]
\tikzstyle{mainedge}=[thick]
\tikzstyle{channel}=[shape=rectangle,draw=blue!50,fill=blue!20]
\node at (8,5) {$\bf{h}_1$};
\node at (8,4) {$\bf{h}_2$};
\node at (8,3) {$\bf{h}_3$};
\node at (8,2.5) {$\vdots$};
\node at (8,1.5) {$\bf{h}_m$};
\node at (10,5.5) {Reference};
\draw[->,thick] (10,5)--(10,1.25);
\node at (-1,6) {States};
\draw[->,thick] (-0.25,6)--(7.25,6);
\node at (0,5.5) {$S_1$};
\node[zero] (R1_1) at (0,5) {};
\node[zero] (R2_1) at (0,4) {};
\node[one]  (R3_1) at (0,3) {};
\node[zero] (Rm_1) at (0,1.5) {};
\node               at (2,5.5) {$S_2$};
\node[one] (R1_2) at (2,5) {}
    edge[lightedge] (R1_1)
    edge[lightedge] (R2_1)
    edge[lightedge] (R3_1)
    edge[lightedge] (Rm_1);
\node[zero] (R2_2) at (2,4) {}
    edge[lightedge] (R1_1)
    edge[lightedge] (R2_1)
    edge[lightedge] (R3_1)
    edge[lightedge] (Rm_1);
\node[one] (R3_2) at (2,3) {}
    edge[mainedge] (R1_1)
    edge[lightedge] (R2_1)
    edge[lightedge] (R3_1)
    edge[lightedge] (Rm_1);
\node[zero] (Rm_2) at (2,1.5) {}
    edge[lightedge] (R1_1)
    edge[lightedge] (R2_1)
    edge[lightedge] (R3_1)
    edge[lightedge] (Rm_1);
\node               at (4,5.5) {$S_3$};
\node[one] (R1_3) at (4,5) {}
    edge[lightedge]  (R1_2)
    edge[lightedge] (R2_2)
    edge[lightedge] (R3_2)
    edge[lightedge] (Rm_2);
\node[one] (R2_3) at (4,4) {}
    edge[lightedge] (R1_2)
    edge[lightedge] (R2_2)
    edge[mainedge] (R3_2)
    edge[lightedge] (Rm_2);
\node[zero] (R3_3) at (4,3) {}
    edge[lightedge] (R1_2)
    edge[lightedge] (R2_2)
    edge[lightedge] (R3_2)
    edge[lightedge] (Rm_2);
\node[zero] (Rm_3) at (4,1.5) {}
    edge[lightedge] (R1_2)
    edge[lightedge] (R2_2)
    edge[lightedge] (R3_2)
    edge[lightedge] (Rm_2);
\node at (5.5,5.5) {$\cdots$};
\node               at (7,5.5) {$S_n$};
\node[one] (R1_n) at (7,5) {}
    edge[lightedge]  (R1_3)
    edge[lightedge] (R2_3)
    edge[lightedge] (R3_3)
    edge[lightedge] (Rm_3);
\node[one] (R2_n) at (7,4) {}
    edge[lightedge] (R1_3)
    edge[mainedge] (R2_3)
    edge[lightedge] (R3_3)
    edge[lightedge] (Rm_3);
\node[zero] (R3_n) at (7,3) {}
    edge[lightedge] (R1_3)
    edge[lightedge] (R2_3)
    edge[lightedge] (R3_3)
    edge[lightedge] (Rm_3);
\node[one] (Rm_n) at (7,1.5) {}
    edge[lightedge] (R1_3)
    edge[lightedge] (R2_3)
    edge[lightedge] (R3_3)
    edge[lightedge] (Rm_3);
\draw [thick,dashed] (-0.5,1) rectangle (7.5,0.5);

\node[zero] (I1) at (0,0.75) {}
    edge[<-,dotted,bend left=45,thick,draw=blue!50] (R1_1);
\node[one] (I2) at (2,0.75) {};
\node[one] (I3) at (4,0.75) {};
\node[one] (In) at (7,0.75) {};

\draw[->,thick] (I1)--(0,0);
\draw[->,thick] (I2)--(2,0);
\draw[->,thick] (I3)--(4,0);
\draw[->,thick] (In)--(7,0);

\draw [thick] (-0.5,0) rectangle (7.5,-0.5);
\node at (3.5,-0.25) {BSC$_\theta$};

\draw [thick,dashed] (-0.5,-0.9) rectangle (7.5,-1.5);
\node[zero] (O1) at (0,-1.25) {};
\node[one] (O2) at (2,-1.25) {};
\node[zero] (O3) at (4,-1.25) {};
\node[one] (On) at (7,-1.25) {};

\draw[->,thick] (0,-0.5)--(O1);
\draw[->,thick] (2,-0.5)--(O2);
\draw[->,thick] (4,-0.5)--(O3);
\draw[->,thick] (7,-0.5)--(On);
\node[align=center] at (8,-1.25) {$\mathbf{X}$};

\draw (I1)--(I2)--(I3)--(In);
\draw (O1)--(O2)--(O3)--(On);

\node [zero,label=right:Reference allele] at (8.5,0.25) {};
\node [one,label=right:Alternative allele] at (8.5,-0.5) {};

\end{tikzpicture}
\end{center}
\caption{
A graphical illustration of HMM for genomes. 
The state space of the hidden states is $\{1,\ldots,m\}$,
where each element corresponds to an index into the reference dataset $\{\mathbf{h}_1,\ldots,\mathbf{h}_m\}$ (each of length $n$). 
A Markov process $\{S_i\}_{i=1}^n$ indicates which reference sequence the user reads the data from at the $i$-th position. For each $i$, $X_i$ differs from the $i$-th position of ${{\bf h}_{S_i}}$ with probability $\theta$, representing noise in the data. BSC$_\theta$: Binary symmetric channel with crossover probability $\theta$. 
}
\label{fig:model}
\end{figure*}


\subsection{An efficient algorithm for HMMs}
\label{sec:algorithm}
In this section, we propose an efficient algorithm to implement the privacy mechanism introduced in Section~\ref{sec:mechanism} for $\pb{\bx}$ based on a hidden Markov model $(\cH,\epsilon,\theta)$ described in the previous section.
The outline of our algorithm is provided in Algorithm~1.

As seen in \eqref{eq:encode}, the privacy mechanism determines the probability of erasing $x_{i}$ mainly based on the probability $\pb{x_{i}|x_{\cK},y_{[i-1]}}$.
By employing a belief propagation approach akin to the well-known forward-backward algorithm~\cite{forward_backward}, we track the computation of $\pb{x_{i}|x_{\cK}=u,y_{[i-1]}}$ for all $u \in \cX^{|\cK|}$ efficiently.  The novelty of our algorithm is that it incorporates the stochasticity of the privacy mechanism in addition to that of the HMM.

First, note that it is sufficient to describe how to compute $\pb{x_i|x_{\cK}=u,y_{[i-1]}}$ for all $u \in \cX^{|\cK|}$ and $i\in [n]$, which fully determines the distribution of $y_1,\dots,y_n$ specified by our privacy mechanism, \Ie
\begin{equation}
\label{eq:track-3}
\pb{y_{i}|x_{i},x_{\cK},y_{[i-1]}} 
 = 
\begin{cases}
    \frac{\min_{u \in \cX^{|\cK|}} \pb{x_{i}|x_{\cK}=u,y_{[i-1]}}}{\pb{x_{i}|x_{\cK},y_{[i-1]}}}, & \text{if }y_{i}=x_{i},\\
    1 - \frac{\min_{u \in \cX^{|\cK|}} \pb{x_{i}|x_{\cK}=u,y_{[i-1]}}}{\pb{x_i|x_{\cK},y_{[i-1]}}}, & \text{if }y_{i}=\ast,\\
    0, & \text{otherwise.}
\end{cases}
\end{equation}


We begin by expressing $\pb{x_{i}|x_{\cK}=u,y_{[i-1]}}$ as
\begin{align}
  \pb{x_{i}|x_{\cK}=u,y_{[i-1]}} 
  & = \sum_{s_{i}} \pb{s_{i}|x_{\cK}=u,y_{[i-1]}} \pb{x_{i}|s_{i},x_{\cK}=u,y_{[i-1]}} \nonumber \\
  & = \sum_{s_{i},s_{{i-1}}}\pb{s_{{i-1}}|x_{\cK}=u,y_{[i-1]}} \pb{{s_i}|s_{{i-1}},x_{\cK}=u} \pb{x_{i}|s_{i}}. \label{eq:track-0}
\end{align}
Note that
\begin{align}
    \pb{{s_i}|s_{{i-1}},x_{\cK}=u}
    &  = \frac{\pb{{s_i},s_{{i-1}},x_{\cK}=u}}{\pb{s_{{i-1}},x_{\cK}=u}} \nonumber\\
    &  = \frac{\pb{s_{{i-1}}|x_{\cK_{i-}}=u_{-}}\pb{s_{{i}}|s_{i-1}}\pb{x_{\cK_{i+}}=u_{+}|s_{i}}}{\pb{s_{{i-1}}|x_{\cK_{i-}}=u_{-}}\pb{x_{\cK_{i+}}=u_{+}|s_{i-1}}} \nonumber \\
    &  = \frac{\pb{s_{{i}}|s_{i-1}}\pb{x_{\cK_{i+}}=u_{+}|s_{i}}}{\pb{x_{\cK_{i+}}=u_{+}|s_{i-1}}} \nonumber\\
     &  = \frac{\pb{s_{{i}}|s_{i-1}}\pb{x_{\cK_{i+}}=u_{+}|s_{i}}}{\sum_{s_i} \pb{s_i|s_{i-1}}\pb{x_{\cK_{i+}}=u_{+}|s_{i}}}, \label{eq:intermediate}
\end{align}
where $\cK_{i-}:=\cK \cap \{1,\ldots,i-1\}$, $\cK_{i+}:=\cK \cap \{i,\ldots,n\}$, $u_{-}$ and $u_{+}$ are corresponding values of $x_{\cK_{i-}}$ and $x_{\cK_{i+}}$ specified by $u$.

As $\pb{x_i|s_i}$ and $\pb{s_i|s_{i-1}}$ are directly given by the HMM, we need only to consider how to compute the two terms $\pb{s_{{i-1}}|x_{\cK}=u,y_{[i-1]}}$ and $\pb{x_{\cK_{i+}}=u_{+}|s_{i}}$.
To simplify our notation, we introduce the following variables to represent these terms:
\begin{align*}
    \psi^{(i)}(u,s_i) &:= \pb{s_i|x_{\cK}=u,y_1,\ldots,y_i},\\
    \gamma^{(i)}(u,s_i) &:= \pb{x_{\cK_{i+}}=u_{+}|s_i}.
\end{align*}
With $\psi^{(i)}(u,s_{i})$ and $\gamma^{(i)}(u,s_i)$ for a given position $i$, 
we can calculate \eqref{eq:track-0} as
\begin{equation}
\label{eq:track-1}
 \pb{x_i|x_{\cK}=u,y_{[i-1]}}  = 
 \sum_{s_{{i-1}}}    \frac{\sum_{s_{i}}\psi^{(i-1)}(u,s_{i-1}) \pb{s_{{i}}|s_{i-1}}\gamma^{(i)}(u,s_i)\pb{x_{i}|s_{i}}}{\sum_{s_i} \pb{s_i|s_{i-1}}\gamma^{(i)}(u,s_i)} .
\end{equation}

First, note that $\gamma^{(i)}(u,s_i)$ can be recursively computed in the same manner as calculating the backward probabilities in the forward-backward algorithm, as described below:

\textbf{Initialization:}
We initialize $\gamma^{(n)}(u,s_n)$ by
\begin{equation}
\label{eq:gamma-initial}
    \gamma^{(n)}(u,s_n) =
    \begin{cases}
        \pb{x_n=u_n|s_n}, & n \in \cK,\\
        1, & n \notin \cK.
    \end{cases}
\end{equation}

\textbf{Iterations:}
For $i=n-1,\ldots,1$, we compute $\gamma^{(i)}(u,s_i)$  as 
\begin{small}
\begin{equation}
\label{eq:gamma-interation}
    \gamma^{(i)}(u,s_i) 
   = 
    \begin{cases}
        \sum_{s_{i+1}} \pb{x_i=u_i|s_i} \pb{s_{i+1}|s_i}  \gamma^{(i+1)}(u,s_{i+1}), & i \in \cK \\
        \sum_{s_{i+1}} \pb{s_{i+1}|s_i} \gamma^{(i+1)}(u,s_{i+1}), & i \notin \cK.
    \end{cases}
\end{equation}
\end{small}

Next, to efficiently compute  $\psi^{(i)}(u,s_i)$ for $i \in [n]$, we analogously adopt the following iterative steps.

\textbf{Initialization:}
$\psi^{(1)}(u,s_1)$ is initialized by
\begin{equation}
\label{eq:psi-initial}
    \psi^{(1)}(u,s_1) \propto  \pb{s_1|x_{\cK}=u}\pb{y_1|s_1, x_{\cK}=u},
\end{equation}
where $\pb{s_1|x_{\cK}=u}$ can be calculated by \eqref{eq:intermediate} given $\gamma^{(1)}(u,s_1)$, and $\pb{y_1|s_1, x_{\cK}=u}$ is given by our mechanism as shown in \eqref{eq:track-3}.

\textbf{Iterations:}
Using Bayes' rule, we can express $\psi^{(i)}(u,s_i)$ as  
\begin{equation}
\label{eq:psi-interation}
     \psi^{(i)}(u,s_i)   = \pb{s_i|x_{\cK}=u,y_{[i]}} 
     \propto \pb{s_i|x_{\cK}=u,y_{[i-1]}} \pb{y_i|s_i, x_{\cK}=u,y_{[i-1]}}, 
\end{equation}
where
\begin{align}
   \pb{s_i|x_{\cK}=u,y_{[i-1]}}  = \sum_{s_{i-1}} \psi^{(i-1)}(u,s_{i-1})\pb{s_i|s_{i-1},x_{\cK}=u},  
\end{align}
and
\begin{align}
     \pb{y_i|s_i, x_{\cK}=u,y_{[i-1]}}  = \sum_{x_i} \pb{x_i|s_i}  \pb{y_i|x_i,x_{\cK}=u,y_{[i-1]}}.
\end{align}

Therefore, $\psi^{(i)}(u,s_i)$ can be computed based on $\psi^{(i-1)}(u,s_{i-1})$. We note that the probability  $\pb{s_i|s_{i-1},x_{\cK}=u}$ can be calculated using $\gamma^{(i)}(u,s_i)$ as shown in \eqref{eq:intermediate}, 
and $\pb{y_i|x_i,x_{\cK}=u,y_{[i-1]}}$ is given by our mechanism as shown in~\eqref{eq:track-3}.
Using this recurrence relation, $\psi^{(i)}(u,s_i)$ for all $i\in [n]$ can be computed.

Analogous to the forward-backward algorithm, our algorithm has polynomial computational complexity of $\cO(nm^{2})$ for a fixed $u$, 
with respect to the sequence length $n$ and the number of reference sequences $m$, for a given $u$.
Clearly, $\min_{u \in \cX^{|\cK|}} \pb{x_{i}|x_{\cK}=u,y_{[i-1]}}$ can be easily obtained once $\pb{x_i|x_{\cK}=u,y_{[i-1]}}$ for all $u$ have been computed.
This overhead involves a factor of $2^{|\cK|}$ in the computational complexity, 
but we expect $|\cK|$ to be a small constant in practice (\Eg less than 10); 
since genotype correlation is predominantly local, the user may apply our mechanism to local regions of the genome of a permissive length, each of which including only a few sensitive positions.

\begin{algorithm}
\begin{algorithmic}[1]
\REQUIRE Genome sequence $\mathbf{X}=(X_1,\ldots,X_n)$ from an HMM with parameters $(\cH,\epsilon,\theta)$, and indices of sensitive positions $\cK \subset [n]$
\ENSURE  Masked genome sequence $\mathbf{Y}=(Y_1,\ldots,Y_n)$, such that $I(X_{\cK};\mathbf{Y})=0$ and $Y_i\in\{X_i,\ast\}$ for all $i\in [n]$

\STATE Initialize $\gamma^{(n)}(u,s_n)$ according to \eqref{eq:gamma-initial}
\FOR{$i=n-1,\ldots,1$} 
    \FOR{$u \in \cX^{|\cK|}$}

        \STATE Compute $\gamma^{(i)}(u,s_i)$ 
        according to \eqref{eq:gamma-interation}
        
        \STATE Compute $\pb{s_i|s_{i-1},x_{\cK}=u}$ 
        according to \eqref{eq:intermediate}
    \ENDFOR
\ENDFOR 

\STATE Initialize $\psi^{(1)}(u,s_1)$ according to \eqref{eq:gamma-initial}

\FOR{$i=2,\ldots,n$} 
  
    \STATE Calculate the erasure probability for $Y_i$ using~\eqref{eq:track-3}
    
      \STATE Generate $Y_i\in \{X_i,\ast\}$ according to the erasure probability


    \FOR{$u \in \cX^{|\cK|}$} 
      
        \STATE Compute $\psi^{(i)}(u,s_i)$ according to \eqref{eq:psi-interation}
    \ENDFOR
    \ENDFOR 
\end{algorithmic}
\caption{Mechanism for hiding  sensitive genotypes in $\bX$}
\label{alg:algorithm}
\end{algorithm}

\section{Simulations}
\label{sec:simulation}
In this section, we provide insights into the empirical performance of our privacy mechanism for hidden Markov models (HMMs) on simulated datasets.
We randomly generated 100 haplotype sequences of length 100, which together with the choices of error probability $\theta$ and crossover probability $\epsilon$ induce $\pb{\bx}$, as described in Section~\ref{sec:HMM}.
For simplicity, we suppose the sensitive position $\mathcal{K}=\{1\}$.

We first illustrate the privacy-utility trade-off of the heuristic window-based erasure approach described in the Introduction.
In particular, this approach erases the first $\omega$ positions of the sequence to hide information about the sensitive position (the first position).
The results are shown in Figure~\ref{fig:window}.
The erasure rate is defined by the size of the erased window over the sequence length, \Ie $\omega/n$ (note $n=100$).
The privacy leakage is measured by the mutual information between the released positions and the sensitive position $X_1$, normalized by the entropy of $X_1$, \Ie $I(X_1;X_{[n]\setminus [\omega]})/H(X_1)$.
We also show the expected erasure rate of our proposed privacy mechanism for comparison, whose privacy leakage is strictly zero by design. We observe that the window-erasure approach requires a high erasure rate (around 0.3) to keep the privacy leakage close to zero, whereas our mechanism achieves a considerably smaller erasure rate (around 0.12) while providing perfect privacy.
On the other hand, choosing a window size for the baseline approach to match the erasure rate of our mechanism leads to a considerable privacy leakage.
%

\begin{figure}[htbp]
\centering
\includegraphics[width=0.75\textwidth]
{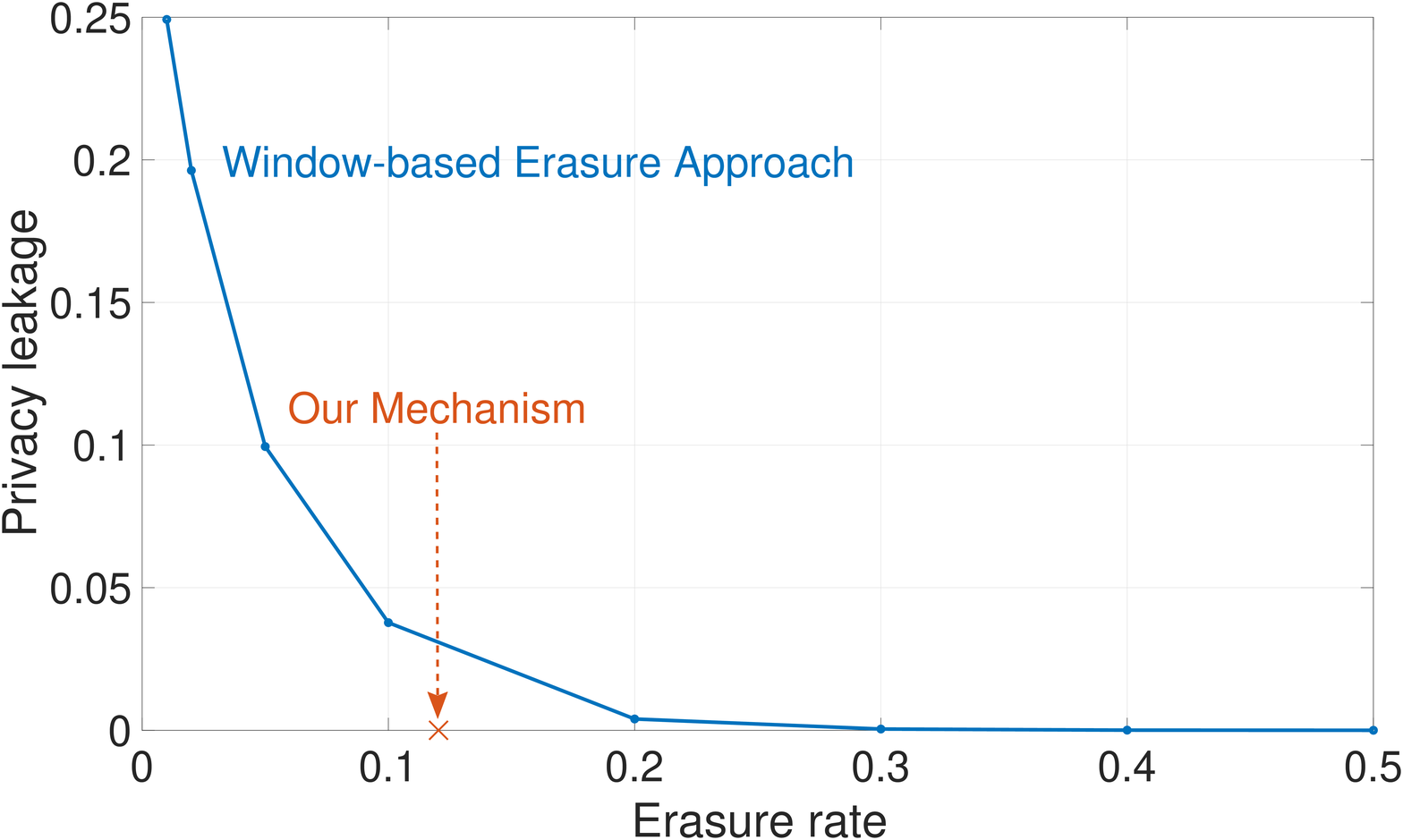}
\caption{
Privacy-utility trade-off of the window-based erasure approach on simulated HMM data with $m=100$, $n=100$, $\cK=\{1\}$, crossover probability $\epsilon = 0.1$ and error probability $\theta = 0.01$. Erasure rate  denotes the size of window that is erased normalized by the sequence length $n$. Privacy leakage denotes the mutual information between the released data and the sensitive symbol normalized by the entropy of the sensitive symbol. 
}
\label{fig:window}
\end{figure}

\begin{figure}[htbp]
\centering
\includegraphics[width=0.75\textwidth]
{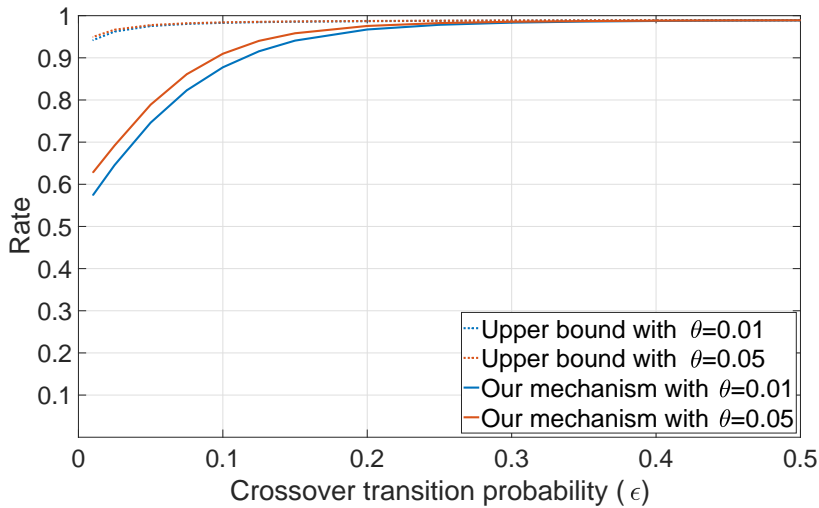}
\caption{
Comparison of our mechanism and the upper bound on simulated HMM data with $m=100$, $n=100$, $\cK=\{1\}$ and different choices of crossover probability $\epsilon$ and error probability $\theta$. 
}
\label{fig:simulation}
\end{figure}

We next evaluate our privacy mechanism over a range of different parameter settings. 
We consider $\theta \in \{0.01, 0.05\}$ and vary $\epsilon$ from 0.01 to 0.5, both of which reflect reasonable ranges of the parameters for the scale of the dataset we simulated.
We provide each instance of $\pb{\bx}$ to our privacy mechanism with $\mathcal{K}=\{1\}$ to calculate its achievable rate $R$ (\Ie one minus the expected erasure rate).
Figure~\ref{fig:simulation} shows the comparison between the rate of our mechanism and the upper bound we derived in Section~\ref{sec:bound}.
The results suggest that the performance of our mechanism shows varying degrees of closeness to the theoretical upper bound depending on the characteristics of the underlying data distribution.
In particular, for higher values of $\epsilon$, representing the regime where the hidden Markov model mixes faster and thus the correlation with the sensitive position decays more quickly, the rate of our mechanism is nearly identical to the upper bound.
On the other hand, for lower values of $\epsilon$, which lead to stronger correlations in the sequence, we observed that the gap between our mechanism and the upper bound can grow considerable large.
Note that this does not necessarily imply that our mechanism achieves a significantly suboptimal performance, given that the upper bound we considered is not tight in general.
We also note that the rate of our mechanism is generally higher when the error probability is larger ($\theta=0.05$ vs $0.01$), which agrees with the intuition that higher levels of noise in the data distribution lower the requirement for hiding sensitive information, thus leading to lower erasure probabilities and higher rates as a result.

\begin{figure}[htbp]
\centering
\includegraphics[width=0.75\textwidth]{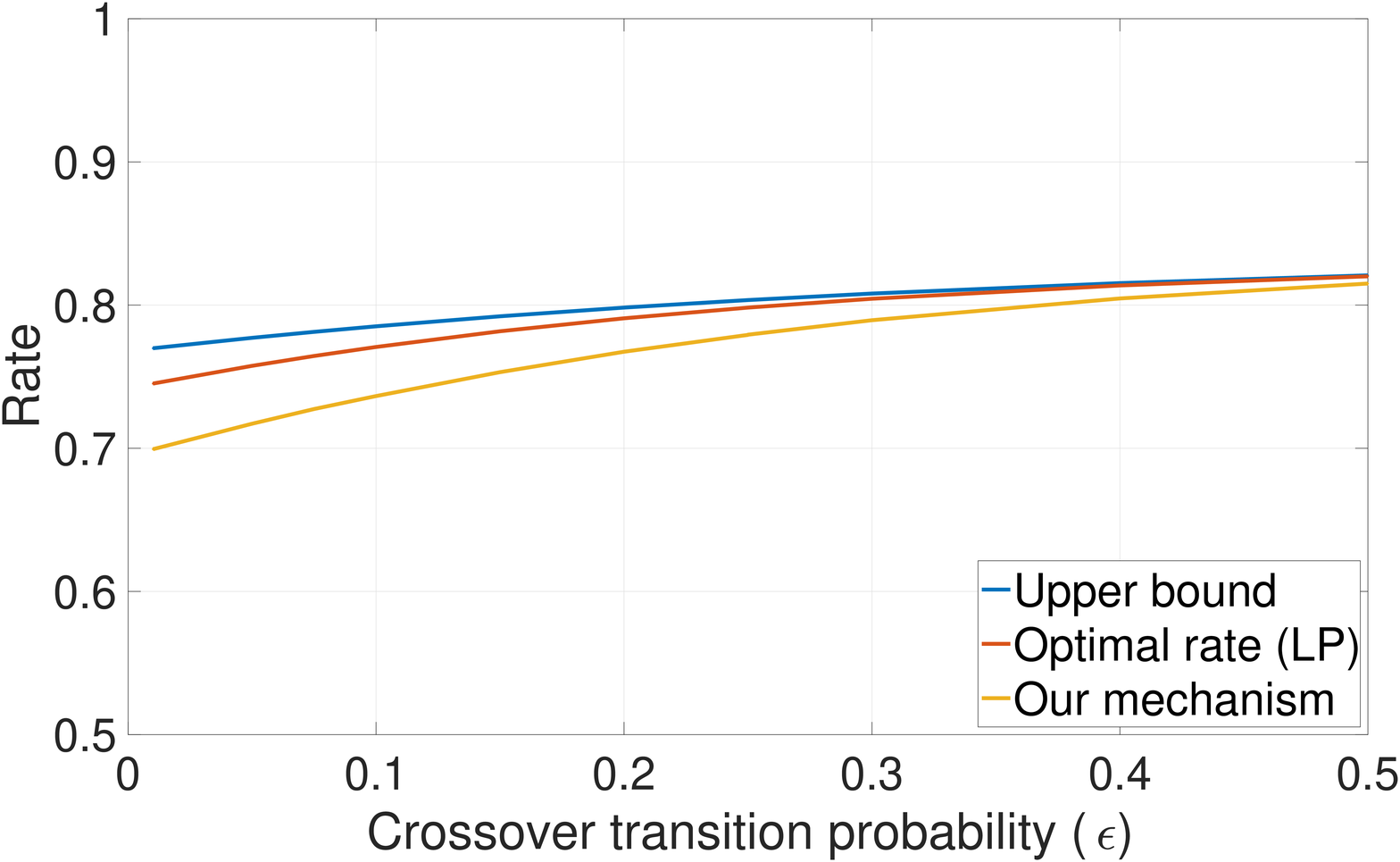}
\caption{
Comparison of our mechanism, the upper bound and the optimal rate based on a linear programming (LP) solution on simulated HMM data with $m=100$, $n=6$, $\cK=\{1\}$ based on a truncated version of the dataset used in Fig.~\ref{fig:simulation}. 
} 
\label{fig:LP}
\end{figure}



To gain further insights into the noticeable gap between the upper bound and our mechanism in the small $\epsilon$ regime, we additionally implemented a linear programming (LP) approach for directly obtaining the optimal mechanism $\wb{\by|\bx}$.
However, since the size of LP grows exponentially with the length of the sequence $n$, we could only evaluate this approach for small problem instances due to numerical instability.
We took the same simulated data as before and truncated each reference haplotype down to the first six positions to obtain a tractable LP instance for this experiment ($n=6$).

The rate comparisons of our privacy mechanism, LP-based optimal mechanism, and the upper bound in this setting are shown in  Fig.~\ref{fig:LP}.
As expected, we observed that the optimal rate lies between the upper bound and the rate of our mechanism, demonstrating that the gap between the optimal rate and the rate of our mechanism is indeed smaller than the ostensible gap suggested by the upper bound.

Taken together, these results suggest that, although the performance of our mechanism is often quite close to the upper bound, the difference between the maximum achievable rate and the rate of our mechanism can vary based on the properties of the data distribution.
We note that it is yet unknown whether there exists a privacy mechanism that can be as efficiently constructed as our mechanism while achieving performance that is closer to the optimal rate.
Closing this performance gap both by devising enhanced privacy mechanisms that achieve higher rates and by developing tighter upper bounds are important directions for future work.



\section{Conclusion and Future Work}
\label{sec:conclusion}

In this paper, we introduced the genotype hiding problem and proposed an information-theoretic privacy mechanism as a solution.
We analyzed the theoretical properties of the mechanism,
and proposed an efficient algorithmic implementation of the mechanism for 
hidden Markov models, a main model of interest for our application in genomics.

It is worth noting that our mechanism does not rule out the possibilities of genotype reconstruction attacks that leverage (i) alternative genetic sequence models and imputation strategies or (ii) a larger set of reference dataset using which HMM parameters could be more accurately estimated. However, our model based on HMMs is consistent with the state-of-the-art techniques for genotype imputation, which is a relatively mature field. In addition, given the high cost of amassing large-scale genomic data, it would be a significant challenge for an attacker to gain access to a larger dataset than those in the public realm. As such, our mechanism could be thought of as providing privacy protection according to the best knowledge of the field.
Our results in Section~\ref{sec:robustness} show that any unforeseen privacy leakage arising from the discrepancies in the data distribution scales gracefully with the relative entropy between the true distribution and the one used by the mechanism.

There are several key directions for future work.
Our work focused on hiding the content of the sensitive positions, yet a potential concern remains regarding information revealed by the choice of sensitive positions $\cK$.
Any approach relying on erasures for privacy protection may inevitably leak information about $\cK$, since preventing such leakage would generally require erasures to be consistently applied throughout the sequence, which is highly costly in terms of utility if only a small fraction is considered sensitive.
An interesting extension of our work is then to relax the faithfulness condition when hiding the positions is deemed important. A promising approach is to re-sample the erased positions from the data distribution as a post-processing step to the mechanism presented in this paper.
That said, we note that in our application setting, $\cK$ is neither necessarily or nor solely decided by the sequence, as it may be determined based on family history of diseases or curated disease associations in public repositories.
Thus, we believe the mechanism presented in this work is directly applicable in many practical scenarios.

Next, although we focused on achieving perfect privacy (with respect to the given data distribution),
it may be useful in practice to consider 
a relaxed notion 
such as local differential privacy~\cite{differential}. 
This may give the user the ability to determine a more desirable trade-off between the level of privacy and the amount of data to be erased.
From an analytical standpoint, this direction would also lead to useful insights about the achievable points along the privacy-utility trade-off curve defined by the genotype-hiding problem with a relaxed notion of privacy, to complement the results in this work.

Furthermore, it would be interesting to 
explore the generalization of our efficient implementation strategies to a broader class of data generative models beyond HMMs,
which may allow similar mechanisms to be employed to protect sensitive data in other domains.

Lastly, we plan to study the performance of our privacy mechanism on real genetic datasets and release the software implementation of our mechanism for the genetics community in the near future.

Growing threats to genetic privacy are necessitating principled strategies for protecting the privacy of individuals while maintaining the utility of data sharing.
Our work illustrates how such a strategy could be designed from an information-theoretic perspective to enable selective disclosure of personal genomic data.
Our methodology is broadly applicable to other data sharing scenarios involving sensitive data with complex correlation structure.
We hope that our work will help spur the development of a wide range of information-theoretic tools for modelling and preserving private genomic information.

\appendices
\section{Proof of Corollary~\ref{corollary}}
\label{appendix:corollary}

We prove the sufficient condition of the optimality holds for the Markov chain case. We give an inductive proof for the sufficient condition by showing that, for a given $x_i$,
\begin{equation}
  u^{\ast} \in \arg\min_{u} \pb{x_i|x_{\cK}=u,y_{[j-1]}}  
\end{equation}
implies
\begin{equation}
    u^{\ast} \in \arg\min_{u} \pb{x_i|x_{\cK}=u,y_{[j]}}
\end{equation}
for $j=1,\ldots,i-1$.
For each $j$, we consider the following two cases ($y_j \neq \ast$ and $y_j = \ast$):
\begin{enumerate}[label=(\arabic*)]
\item If $y_j \neq \ast$, then we have 
	\begin{align}
	\pb{x_i|x_{\cK},y_{[j]}} 
		& = \sum_{x_j} \pb{x_j|x_{\cK},y_{[j]}} \pb{x_i|x_{\cK},y_{[j]},x_j} \nonumber \\
		& \utag{a}{=}  \mathbbm{1}\{x_j=y_j\}  \pb{x_i|x_{\cK},y_{[j]},x_j} \nonumber \\
		& \utag{b}{=} \mathbbm{1}\{x_j=y_j\}  \pb{x_i|x_j},
	\end{align}
	where \uref{a} follows because $Y_j$ can either be $X_j$ or $\ast$, and \uref{b} follows from Markovity. In this case, $\arg\min_{u} \pb{x_i|x_{\cK}=u,y_{[j]}}$ is indeed independent of $u$, which means 
	\begin{equation}
	\arg\min_{u} \pb{x_i|x_{\cK}=u,y_{[j]}} = |\cX|,    
	\end{equation}
	 so the statement is trivially true. 

\item If $y_j = \ast$, then we have 	
		\begin{align}
	 \pb{x_i|x_{\cK},y_{[j]}} 
		& = \sum_{x_j} \pb{x_j|x_{\cK},y_{[j]}} \pb{x_i|x_{\cK},y_{[j]},x_j} \nonumber \\
		& \utag{a}{=} \sum_{x_j} \pb{x_j|x_{\cK},y_{[j]}} \pb{x_i|x_j} \nonumber \\
		& \utag{b}{\propto} \sum_{x_j} \left\{ 
			\pb{x_j|x_{\cK},y_{[j-1]}} - \min_{u} \pb{x_j|x_{\cK}=u,y_{[j-1]}}\right\} \pb{x_i|x_j} \nonumber \\
		& = \sum_{x_j} 
			\pb{x_i|x_j} \pb{x_j|x_{\cK},y_{[j-1]}}  - \sum_{x_j} \pb{x_i|x_j}\min_{u} \pb{x_j|x_{\cK}=u,y_{[j-1]}} \nonumber \\
			& =  \pb{x_i|x_{\cK},y_{[j-1]}}  - \sum_{x_j} \pb{x_i|x_j}\min_{u} \pb{x_j|x_{\cK}=u,y_{[j-1]}},	\label{eq:exp-markov-1}
	\end{align}
	where \uref{a} follows from Markovity, and \uref{b} follows from Bayes's rule and our privacy mechanism \eqref{eq:encode}. Since the second term of the right-hand side in \eqref{eq:exp-markov-1} is independent of $x_{\cK}$, we obtain
	\begin{equation}
	 	\arg\min_{u} \pb{x_i|x_{\cK}=u,y_{[j]}} 
	 	 = \arg\min_{u} \pb{x_i|x_{\cK}=u,y_{[j-1]}}.
	 \end{equation}
\end{enumerate}

For both cases, we have verified that the sufficient condition holds, which completes the proof.

\section{Proof of Lemma~\ref{lemma:deterministic}}
\label{appendix:lemma-deterministic}
We will prove \eqref{eq:NP-determine} by induction. First, consider the base case:
\begin{equation}
\wb{y_{o_1}=\ast|x_{o_1},x_{\cK}} = 1 - \frac{\min_{x_{\cK}} \pb{x_{o_1}|x_{\cK}}}{\pb{x_{o_1}|x_{\cK}}}.
\end{equation} 
From the previous discussion, we know that if $o_i \in \cK$, then 
\begin{equation}
	\wb{y_{o_i}=\ast|x_{o_i},x_{\cK},y_{o_{[i-1]}}} = 1,
\end{equation}
so without loss of generality, we assume that 
\begin{equation}
    o_1 \notin \cK = \{m+1,\ldots,m+k\}.
\end{equation}
Since
\begin{equation}
    x_{\cK} = \left\{\sum_{i: i \in S_j} b_{i,j}:j \in [k] \right\},
\end{equation}
and 
\begin{equation}
    x_{o_1} = \left\{b_{o_1,j}: o_1 \in S_j \right\}
\end{equation}
by definition,  we can see that if there exists some $j$ such that $S_j=\{o_1\}$, then $b_{o_1,j} \in x_{\cK}$ and $b_{o_1,j} \in x_{o_1}$. In this case, we can always find some assignments such that
\begin{equation}
	\min_{x_{\cK}} \pb{x_{o_1}|x_{\cK}} = 0,
\end{equation}
implying that 
\begin{equation}
\wb{y_{o_1}=\ast|x_{o_1},x_{\cK}} = 1.
\end{equation}
If there is no $j$ such that $S_j=\{o_i\}$, each $\sum_{i: i \in S_j} b_{i,j}$ constituting $x_{\cK}$ is a binary summation of some $b_{o_1,j}$ and (independent) random bits $b_{i,j}$ such that $i \neq o_1$, where the latter render the result uniformly random.
This means that $X_{\cK}$ is independent of $X_{o_1}$, and thus we have 
\begin{equation}
\wb{y_{o_1}=\ast|x_{o_1},x_{\cK}}  = 1 - \frac{\min_{x_{\cK}} \pb{x_{o_1}|x_{\cK}}}{\pb{x_{o_1}|x_{\cK}}} \\
 = 1 - \frac{\min_{x_{\cK}} \pb{x_{o_1}}}{\pb{x_{o_1}}} = 0,
\end{equation} 
for all $x_{o_1}$ and $x_{\cK}$.

Assume the statement is true for $o_1,\ldots,o_{i-1}$. Then for $o_i$, note that
\begin{equation}
\label{eq:NP-ep1}
\begin{aligned}
& \pb{x_{o_i}|x_{\cK},y_{o_{[i-1]}}}  = \frac{
\pb{x_{o_i}|x_{\cK}}\pb{y_{o_{[i-1]}}|x_{o_i},x_{\cK}}}{\pb{y_{o_{[i-1]}}|x_{\cK}}}. 
\end{aligned}
\end{equation} 
By letting 
\begin{equation}
\label{eq:preceding-non-erasure}
	\tilde{\cE_i} = \left\{o_j: y_{o_j} \neq \ast, j \leq i-1\right\},
\end{equation}
\eqref{eq:NP-ep1} can be written as
\begin{equation}
\pb{x_{o_i}|x_{\cK},y_{o_{[i-1]}}} = \frac{
\pb{x_{o_i}|x_{\cK}}\pb{x_{\tilde{\cE_i}}|x_{o_i},x_{\cK}}}{\pb{x_{\tilde{\cE_i}}|x_{\cK}}} 
 = \pb{x_{o_i}|x_{\tilde{\cE_i}},x_{\cK}},
\end{equation} 
because of the inductive assumption that the decisions whether to erase $y_{o_1},\ldots,y_{o_{i-1}}$ are deterministic.

Hence, we have
\begin{align}
 \wb{y_{o_i}=\ast|x_{o_i},x_{\cK},y_{o_{[i-1]}}} 
& = 1 - \frac{\min_{x_{\cK}} \pb{x_{o_i}|x_{\cK},y_{o_{[i-1]}}}}{\pb{x_{o_i}|x_{\cK},y_{o_{[i-1]}}}} \nonumber\\
& = 1 - \frac{\min_{x_{\cK}} \pb{x_{o_i}|x_{\tilde{\cE_i}},x_{\cK}}}{\pb{x_{o_i}|x_{\tilde{\cE_i}},x_{\cK}}}.
\end{align}
Analogous to our argument for the base case, if there exists some $j$ such that $S_j \subseteq \tilde{\cE_i} \cup \{o_i\}$, then one can determine $b_{o_i,j} \in x_{o_1}$ from $x_{\tilde{\cE_i}}, x_{\cK}$, and thus 
\begin{equation}
	\min_{x_{\cK}} \pb{x_{o_i}|x_{\tilde{\cE_i}},x_{\cK}} = 0,
\end{equation}
implying that 
\begin{equation}
\wb{y_{o_i}=\ast|x_{o_i},x_{\cK},y_{o_{[i-1]}}} = 1.
\end{equation}
If there is no such $j$, each $x_{j}$ for $j \in \cK$ is the binary summation of some $b_{o_i,j} \in x_{o_i}$ and some independent random bits $b_{i',j}$ such that $i' \neq o_i$, which again guarantees that $X_{\cK}$ is independent of $X_{o_i}$ conditioning on $X_{\tilde{\cE_i}}$. Thus, we have 
\begin{equation}
\wb{y_{o_i}=\ast|x_{o_i},x_{\cK},y_{o_{[i-1]}}} = 
1 - \frac{\min_{x_{\cK}} \pb{x_{o_i}|x_{\tilde{\cE_i}}}}{\pb{x_{o_i}|x_{\tilde{\cE_i}}}}  = 0,
\end{equation} 
for all $x_{o_i}$, $x_{\cK}$ and $y_{o_{[i-1]}}$, which completes the inductive proof.


\section{Proof of Theorem~\ref{theorem:np}}
\label{appendix:hitting-set}

First, let us show that $e^{\ast} \geq h^{\ast}$ by showing that $E_{\pi}$ is a hitting set for any order $\pi$, \Ie $E_{\pi} \cap S_j \neq \emptyset$ for all $j \in [k]$. We prove it by contradiction. Suppose that there exists some $S_j$ such that $E_{\pi} \cap S_j = \emptyset$, which implies that $S_j \subseteq [m]\backslash E_{\pi}$ for some $j$. Assume that $S_j = \{i_1,\ldots,i_t\}$, and $i_t$ is the last index visited that specified by the given order $\pi$. Then, when we run our mechanism for $i_t$, since $i_1,\ldots,i_{t-1}$ are all visited and not erased, by recalling the proof of Lemma~\ref{lemma:deterministic}, we know that $\tilde{\cE}_{i_t} \supseteq \{i_1,\ldots,i_{t-1}\}$, so we have $S_j \subseteq \tilde{\cE}_{i_t} \cup \{i_{t}\}$. It means that $y_{i_t}$ is erased or $i_t \in E_{\pi}$, which contradicts with our assumption $E_{\pi} \cap S_j = \emptyset$. 

Next, we show that $e^{\ast} \leq h^{\ast}$ by showing that for any given hitting set $V$, there exists an order $\pi$ such that $|E_{\pi}| \leq |V|$.  Suppose $V$ is a hitting set and $|V|=h$, \Ie $V \cap S_j \neq \emptyset$ for all $j \in [k]$. 
Consider an order $\pi$ such that $o_i \notin V \cup [m+1:m+k]$ for $i \leq m-h$ and $o_i \in V$ for $i \in [m-h+1:m]$, \Ie visiting indices in the complementary of $T$ before attaining $V$.  When we visit $o_i$ such that $i \leq m-h$ (or $o_i \in [m]\backslash V$), by the assumption that $V \cap S_j \neq \emptyset$ for all $j$, we know that there exists some index $t_j \in S_j \cap V$ for each $j$. By recalling the definition \eqref{eq:preceding-non-erasure}, we know that $\tilde{\cE}_{i} \supseteq [m]\backslash V$, so $t_j \notin \tilde{\cE}_{i}$. Note that $t_j \in V$ while $o_i \notin V$, so $t_j \notin \tilde{\cE}_{i} \cup \{o_i\}$. Hence, we know that $y_{o_i}$ is not erased, or $o_i \notin E_{\pi}$ from the proof of Lemma~\ref{lemma:deterministic}. Since $o_i \notin E_{\pi}$ for $i \leq m-h$ given this particular order $\pi$, we have $|E_{\pi}| \leq h = |V|$, which completes the proof.

\section{Proof of Theorem~\ref{thm:leakage}}
\label{appendix:leakage}
From \eqref{eq:robust-intermediate}, we have 
\begin{align}
I(\pb{x_{\cK}};\pb{\bf{y}}) 
		&  = D(\pb{x_{\cK},\bf{y}}||\qb{x_{\cK},\bf{y}}) - D(\pb{x_{\cK}}||\qb{x_{\cK}})  - D(\pb{\bf{y}}||\qb{\bf{y}}), \label{eq:appendix-robust-eq1}
\end{align}
and it remains to show that the right-hand side is bounded above by $D(\pb{\bf{x}}||\qb{\bf{x}})$.

By applying the chain rule of relative entropy, we have
\begin{equation}
	D(\pb{\bf{x},\bf{y}}||\qb{\bf{x},\bf{y}}) = D(\pb{\bf{x}}||\qb{\bf{x}})+D(\pb{\bf{y}|\bf{x}}||\qb{\bf{y}|\bf{x}}),
\end{equation}
and 
\begin{equation}
	D(\pb{\bf{x},\bf{y}}||\qb{\bf{x},\bf{y}})  = D(\pb{x_{\cK},\bf{y}}||\qb{x_{\cK},\bf{y}}) 
	  + D(\pb{x_{[n]\backslash\cK}|x_{\cK},\bf{y}}||\qb{x_{[n]\backslash\cK}|x_{\cK},\bf{y}}).
\end{equation}
From these equations, we obtain
\begin{multline}
\label{eq:appendix-robust-eq2}
D(\pb{x_{\cK},\bf{y}}||\qb{x_{\cK},\bf{y}}) \\
	=
	D(\pb{\bf{x}}||\qb{\bf{x}}) + D(\pb{\bf{y}|\bf{x}}||\qb{\bf{y}|\bf{x}}) - D(\pb{x_{[n]\backslash\cK}|x_{\cK},\bf{y}}||\qb{x_{[n]\backslash\cK}|x_{\cK},\bf{y}}).
\end{multline}
By substituting \eqref{eq:appendix-robust-eq2} in \eqref{eq:appendix-robust-eq1}, we have 
\begin{align}
I(\pb{x_{\cK}};\pb{\bf{y}}) 
	& = D(\pb{\bf{x}}||\qb{\bf{x}})+D(\pb{\bf{y}|\bf{x}}||\qb{\bf{y}|\bf{x}}) - D(\pb{x_{\cK}}||\qb{x_{\cK}}) \nonumber\\
	& ~~~ - D(\pb{x_{[n]\backslash\cK}|x_{\cK},\bf{y}}||\qb{x_{[n]\backslash\cK}|x_{\cK},\bf{y}})  - D(\pb{\bf{y}}||\qb{\bf{y}}) \nonumber \\
	& \utag{a}{\leq} D(\pb{\bf{x}}||\qb{\bf{x}})+D(\pb{\bf{y}|\bf{x}}||\qb{\bf{y}|\bf{x}}) \nonumber \\
	& \utag{b}{=} D(\pb{\bf{x}}||\qb{\bf{x}}),
\end{align}
where \uref{a} follows from the non-negativity of relative entropy, \uref{b} follows from the assumption $\qb{\bf{y}|\bf{x}}= \pb{\bf{y}|\bf{x}}= w(\bf{y}|\bf{x})$.

\section*{Acknowledgment}
The work of F.~Ye and S.~EL~Rouayheb was supported in part by NSF Grant CCF 1817635. H.~Cho is funded by NIH DP5 OD029574-01 and by Eric and Wendy Schmidt through the Schmidt Fellows program at Broad Institute.




\end{document}